\documentclass[12pt,a4paper]{article}
\usepackage{hyperref}
\usepackage{amsthm,amsmath,amssymb}
\usepackage{verbatim}
\usepackage{graphicx}
\usepackage{url}
\usepackage{xcolor}
\usepackage{float}
\usepackage[T1]{fontenc}
\usepackage{cite}

\usepackage[headsepline,nouppercase]{scrpage2}
\clearscrheadfoot
\ihead{\headmark}
\ohead{\pagemark}
\automark{section}
\setheadsepline{0.1pt}

\theoremstyle{theorem}
\newtheorem{theorem}{Theorem}[section]
\newtheorem{lemma}[theorem]{Lemma}
\newtheorem{proposition}[theorem]{Proposition}
\newtheorem{corollary}[theorem]{Corollary}

\theoremstyle{definition}
\newtheorem{definition}[theorem]{Definition}

\theoremstyle{remark}
\newtheorem{remark}[theorem]{Remark}

\numberwithin{equation}{section}

\newcommand{\IN}{\mathbb{N}}

\newcommand{\subf}{\mathrm{sf}}
\newcommand{\kxs}{\mathrm{K4\times S5}}
\newcommand{\sxs}{\mathrm{S4\times S5}}
\newcommand{\Kfour}{\mathrm{K4}}
\newcommand{\Sfour}{\mathrm{S4}}
\newcommand{\Sfive}{\mathrm{S5}}

\newcommand{\ssl}{\mathrm{SSL}}
\newcommand{\SSL}{\mathrm{SSL}}
\newcommand{\EXPSPACE}{\mathrm{EXPSPACE}}
\newcommand{\ESPACE}{\mathrm{ESPACE}}

\makeindex

\title
{$\EXPSPACE$-Completeness of the Logics $\kxs$ and $\sxs$ and the Logic of Subset Spaces, \\
Part 1: $\ESPACE$-Algorithms}

\author{Peter Hertling and Gisela Krommes\\
Fakult\"at f\"ur Informatik \\
Universit\"at der Bundeswehr M\"unchen \\
85577 Neubiberg, Germany \\[2mm]
Email: peter.hertling@unibw.de, gisela.krommes@unibw.de}

\date{\today}

\begin{document}

\maketitle

\begin{abstract}
It is known that the satisfiability problems of the product logics $\kxs$ and $\sxs$ and of the logic $\ssl$ of subset spaces are in $\mathrm{N2EXPTIME}$.
We improve this upper bound for the complexity of these problems by presenting $\mathrm{ESPACE}$-algorithms for these problems.
In another paper we show that these problems are $\EXPSPACE$-hard. This shows that all three problems are $\EXPSPACE$-complete.
\end{abstract}

\bigskip

\noindent{\bf Keywords:}
bimodal product logics, subset space logic, satisfiability problem, complexity theory, $\EXPSPACE$-completeness

\maketitle

\section{Introduction}
One of the fundamental complexity-theoretic results about logic is Cook's theorem
which says that the satisfiability problem for Boolean formulas is $\mathrm{NP}$-complete~\cite{Cook1971}.
Since then the complexity of many other logics has been analysed.
In this article we are concerned with the bimodal product logics $\kxs$ and $\sxs$
and with the subset space logic $\ssl$, a bimodal logic as well.
To the best of our knowledge, the complexity of $\kxs$, of $\sxs$, and of $\ssl$ were open problems.
The main results of this article can be summarized in the following theorem.

\begin{theorem}
	The logics
	$\kxs$, $\sxs$, and $\ssl$ are in $\ESPACE$.
\end{theorem}

Actually, we are considering the satisfiability problems of these three logics, and we are going to show that the satisfiability problems of these logics are in $\ESPACE$. Of course, this assertion is equivalent to the theorem above because $\ESPACE$ is closed under complements. In another paper \cite{HK2019-2} we show that these problems are $\EXPSPACE$-hard under logspace reduction. Both results together imply that all three logics are $\EXPSPACE$-complete under logspace reduction.

Let us recap the history of the questions and results concerning the complexity of these problems. The following text is almost identical with a corresponding text in ~\cite{HK2019-2}.

In \cite[Question 5.3(i)]{marx1999complexity} Marx posed the question what the complexity of the bimodal logic $\sxs$ is. 
This question is restated and extended to the logic $\kxs$ in~\cite[Problem 6.67, Page 334]{Kurucz2003}.
There it is also stated that
``M. Marx conjectures that these logics are also EXPSPACE-complete''.
That it is desirable to know the complexity of $\ssl$ and similar logics is mentioned by Parikh, Moss, and Steinsvold in \cite[Page 30]{parikh2007topology}
and by Heinemann in \cite[Page 153]{heinemann2016augmenting} and in \cite[Page 513]{heinemann2016subset}. 

For the complexity of the satisfiability problems of the logics $\kxs$ and $\sxs$ the best upper bound known is $\mathrm{N2EXPTIME}$ \cite[Theorem 5.28]{Kurucz2003}, that is, they can be solved by a nondeterministic Turing machine working in doubly exponential time. The best lower bound known for the satisfiability problems of these two logics is $\mathrm{NEXPTIME}$-hardness \cite[Theorem 5.42]{Kurucz2003}; compare also \cite[Table 6.3, Page 340]{Kurucz2003}.
It is known as well that for any $\ssl$-satisfiable formula there exists a cross axiom model of at most doubly exponential size~\cite[Section 2.3]{Dabrowski1992}. This shows that the complexity of the satisfiability problem of $\ssl$ is in $\mathrm{N2EXPTIME}$ as well. The best lower bound known for $\ssl$ is $\mathrm{PSPACE}$-hardness ~\cite{krommes2003,krommes2003new}.

In this paper we improve the upper bound $\mathrm{N2EXPTIME}$ for the satisfiability problems of these three logics to $\ESPACE$. In another paper~\cite{HK2019-2} we show a matching lower bound by showing that these problems are $\EXPSPACE$-hard. This shows that they are $\EXPSPACE$-complete. Thus, Marx's conjecture for $\kxs$ and $\sxs$ stated above is true.

In Section~\ref{section:combined-modal-logics} we introduce the bimodal logics $\kxs$, $\sxs$, and $\ssl$. Actually, we restrict ourselves to defining only those notions concerning these logics that we need. First the syntax of bimodal formulas is defined, then various kinds of models are presented, and then we define when a bimodal formula is $X$-satisfiable, for $X\in \{\kxs,\sxs,\ssl\}$.
In Section~\ref{section:overview} we formulate a more precise version of the main theorem, that is, a stronger upper bound for the complexity of the satisfiability problems of these three logics than just $\ESPACE$. And we give an overview of the proof. In the following sections we do some more preparations, present the algorithms, prove their correctness, and prove the claimed upper bounds for the space needed by these algorithms. 
We present recursive decision algorithms for these problems that are based on certain kinds of tableaux. We will construct tableaux not as usual brick by brick. Instead we shall use prefabricated parts that we call ``tableau-clouds'' and that are somewhat similar to mosaics~\cite{nemeti1995decidable}. 
Our recursive algorithms are similar to the recursive algorithm of  Ladner~\cite{ladner1977computational} for the modal logic $\mathrm{S4}$. We would like to point out that Section~\ref{section:relations}, in particular Subsection~\ref{subsection:maximum-chain-length}, contains some general combinatorial observations on certain binary relations that may be of interest elsewhere as well.

Let us end this introduction by mentioning some complexity-theoretic notions that will be used. The required notions from logic will be introduced in Section~\ref{section:combined-modal-logics}.
First, as usual $\IN=\{0,1,2,\ldots\}$ is the set of natural numbers, that is, of non-negative integers.
An \emph{alphabet} is a finite, nonempty set. For an alphabet $\Sigma$ let $\Sigma^*$ be the set of all finite strings over $\Sigma$. A \emph{language} is any subset $L\subseteq \Sigma^*$, where $\Sigma$ is any alphabet. For a function $s:\IN\to\IN$ we say that a language $L$ \emph{can be decided in space $O(s)$} if there exists a deterministic Turing machine that decides $L$ in space $O(s)$; for the precise definition of what this means the reader is referred to \cite{Papadimitriou1994} or to any other textbook on complexity theory. The following two complexity classes have already been mentioned.
\begin{itemize}
\item
$\EXPSPACE$ is the set of languages that can be decided by a deterministic Turing machine in space $2^{p(n)}$ for some polynomial $p$. 
\item
$\ESPACE$ is the set of languages that can be decided by a deterministic Turing machine in space $2^{c\cdot n+c}$, for some constant $c\in \IN$, that is, the exponent is linear. 
\end{itemize}
Note that in order to speak about the complexity of a decision problem one should encode the instances of the decision problem by strings. In this way one gets a language.
At first sight it might seem surprising that here we establish $\ESPACE$ as an upper bound and in another paper \cite{HK2019-2} $\EXPSPACE$-hardness as a lower bound. But this is not a contradiction because the complexity class $\ESPACE$ is not closed under reduction, neither reductions running in polynomial time nor those running in logarithmic space.

\section{Definition of the Satisfiability Problems of the logics $\kxs$, $\sxs$, and $\ssl$}
\label{section:combined-modal-logics}

In the first subsection of this section we define the syntax of bimodal formulas. Then we introduce various kinds of models. Finally, we define $X$-satisfiability of bimodal formulas, for each $X\in\{\kxs,\sxs,\ssl\}$.

\subsection{Bimodal Formulas}

Bimodal formulas are defined just like Boolean formulas but with two additional unary (modal) operators, that we write as $\Box$ and as $K$.
For the aimed complexity proofs it is convenient to define the syntax in such a way that propositional variables are represented as $x\,binary$ where $binary$ is some binary number without leading zeros. The set of well-formed bimodal formulas $\mathcal{L}$ is generated by a context-free grammar. 

\begin{definition}[Syntax of Bimodal Formulas] \label{def:syntax}
	The set $\mathcal{L}$ of well-formed \emph{bimodal formulas} is recursively generated using the following Backus-Naur grammar:	 
	\[\begin{array}{lll} 
	\varphi	&::=	& \neg\varphi \mid (\varphi\wedge\varphi) \mid K\varphi \mid \Box\varphi \mid \langle var\rangle\\
	\langle var\rangle	&::=	&x0 \mid x1 \mid x1\langle binstring\rangle\\
	\langle binstring\rangle &::= & 0 \mid 1 \mid 0\langle binstring\rangle \mid 1\langle binstring\rangle
	\end{array} \]
\end{definition}

The set $AT$ of propositional variables in $\mathcal{L}$ is defined by 
$$AT:=\{w\in\mathcal{L}\mid x \text{ is prefix of } w\}.$$ 	
We also need some formulas of special type. For a modal operator $\circ\in\{K,\Box\}$ we define the set
\[\mathcal{L}_\circ := \{\psi \in \mathcal{L} \mid (\exists \chi \in \mathcal{L}) \ \psi = \circ \chi\}\]
We adopt standard abbreviations for additional propositional connectives and the dual modal operators: $(\varphi\vee\psi):=\neg(\neg\varphi\wedge\neg\psi)$, $(\varphi\to\psi):=(\neg\varphi\vee\psi)$, 
$(\varphi\leftrightarrow\psi):=((\neg\varphi\vee\psi)\wedge(\neg\psi\vee\varphi))$,
$L\varphi:=\neg K\neg\varphi$ and $\Diamond\varphi:=\neg\,\Box\neg\varphi$.
We will omit brackets whenever there is no danger that this might lead to confusion.
We introduce some further syntactical concepts and notions:
\begin{definition}[Subformula]
	The set $\subf(\varphi)$ of {\em subformulas} of a bimodal formula $\varphi$ is
	defined as usual by recursion:
	\begin{eqnarray*}
		\subf(A) &:=& \{A\} \text{ for } A\in AT, \\
		\subf(\neg \varphi) &:=& \{\neg \varphi\} \cup\subf(\varphi), \\
		\subf((\varphi \wedge \psi)) &:=& \{(\varphi \wedge \psi)\} \cup \subf(\varphi) \cup \subf(\psi), \\
		\subf(\Box \varphi) &:=& \{\Box \varphi\} \cup \subf(\varphi), \\
		\subf(K \varphi) &:=& \{K \varphi\} \cup \subf(\varphi).
	\end{eqnarray*}	
\end{definition}

\subsection{Several Kinds of Models for Bimodal Logics}
\label{subsection:models}

All three logics considered, $\kxs$, $\sxs$, and $\ssl$, can be considered as combinations of either $\mathrm{K4}$ or $\mathrm{S4}$ with $\mathrm{S5}$.
The logic $\kxs$ is defined as the logic of $\kxs$-product frames, and $\sxs$ is defined as the logic of $\sxs$-product frames, defined as follows.

\begin{definition}[$\kxs$- and $\sxs$-Product Models]
	\label{definition:S4S5product} 
	\begin{enumerate}
		\item
		\sloppy{
		A K4-\emph{frame} is a pair $(W,R_{\Diamond})$
		such that $W$ is a non-empty set and $R_{\Diamond}\subseteq W\times W$ is a transitive relation on $W$. 
		}
		
		An S4-\emph{frame} is a pair $(W,R_{\Diamond})$
		such that $W$ is a non-empty set and $R_{\Diamond}\subseteq W\times W$ is a preorder on $W$,
		that is a reflexive and transitive relation. 
		
		An S5-\emph{frame} is a pair $(W,R_{L})$ such that $W$ is a non-empty set and $R_L\subseteq W\times W$ is an equivalence relation on $W$,
		that is a reflexive, transitive and symmetric relation.
		\item
		Let $X\in\{\Kfour,\Sfour\}$.
		Let $F_1:=(W_1,R_{\Diamond})$ be some $X$-frame,
		$F_2:=(W_2,R_{L})$ be some  S5-frame.
		Then the {\em product $F_1\times F_2$} is the triple
		\[ F:= (W_1\times W_2, \, \stackrel{\Diamond}{\to}, \, \stackrel{L}{\to})\]
		where $\stackrel{\Diamond}{\to}$ and $\stackrel{L}{\to}$ are the binary relations
		on $W_1 \times W_2$ defined by			
		\begin{eqnarray*}
			(v_1,v_2) \stackrel{\Diamond}{\to} (w_1,w_2) 
			&\iff& v_1R_\Diamond w_1  \mbox{ and } v_2 = w_2, \\
			(v_1,v_2) \stackrel{L}{\to} (w_1,w_2) 
			&\iff& v_2R_L w_2  \mbox{ and } v_1 = w_1,
		\end{eqnarray*}
		for all $(v_1,v_2), (w_1,w_2) \in W_1\times W_2$.
		Any such product is called an $X\times\Sfive$-\emph{product frame}.
		\item
		Let $X\in\{\Kfour,\Sfour\}$.
		Then an {\em $X\times\Sfive$-product model} is a
		quadruple $(W, \stackrel{\Diamond}{\to}, \stackrel{L}{\to}, \sigma)$
		such that the triple $(W, \stackrel{\Diamond}{\to}, \stackrel{L}{\to})$ is an
		$X\times\Sfive$-product frame and 
		$$\sigma:AT\to\mathcal{P}(W)$$
		is a function mapping proposition letters to subsets of $W$.
	\end{enumerate}
\end{definition}

Let $X\in\{\Kfour,\Sfour\}$.
Note that the relation $\stackrel{\Diamond}{\to}$ in an $X\times\Sfive$-product frame is automatically transitive and in the case of $X=\Sfour$ even a preorder and that the relation $\stackrel{L}{\to}$ in a product frame is automatically an equivalence relation.
In diagrams we will usually depict the relation $\stackrel{\Diamond}{\to}$ as the `vertical' relation and the relation $\stackrel{L}{\to}$ as the `horizontal' relation, as in Figure~\ref{figure:left-and-right-commutativity}.
Note that it is obvious that any $X\times\Sfive$-product frame, for $X\in\{\kxs,\sxs\}$, satisfies the following two properties:
\begin{itemize}
\item \emph{left commutativity}:
		$\forall w\forall u \forall u' \left( 
		(w \stackrel{\Diamond}{\to} u \wedge u \stackrel{L}{\to} u')
		\rightarrow
		\exists w' (w \stackrel{L}{\to} w' \wedge w' \stackrel{\Diamond}{\to} u')		 
		\right)$, 
\item  \emph{right commutativity}:
		$\forall w\forall w'\forall u' \left(
		(w \stackrel{L}{\to} w' \wedge w' \stackrel{\Diamond}{\to} u')
		\rightarrow
		\exists u (w \stackrel{\Diamond}{\to} u \wedge u \stackrel{L}{\to} u')		 
		\right)$.
\end{itemize}
These properties are illustrated in Figure~\ref{figure:left-and-right-commutativity}, essentially copied from \cite{Kurucz2003}. 
\begin{figure}[h]
   \centering			
	\includegraphics[width=0.6\linewidth]{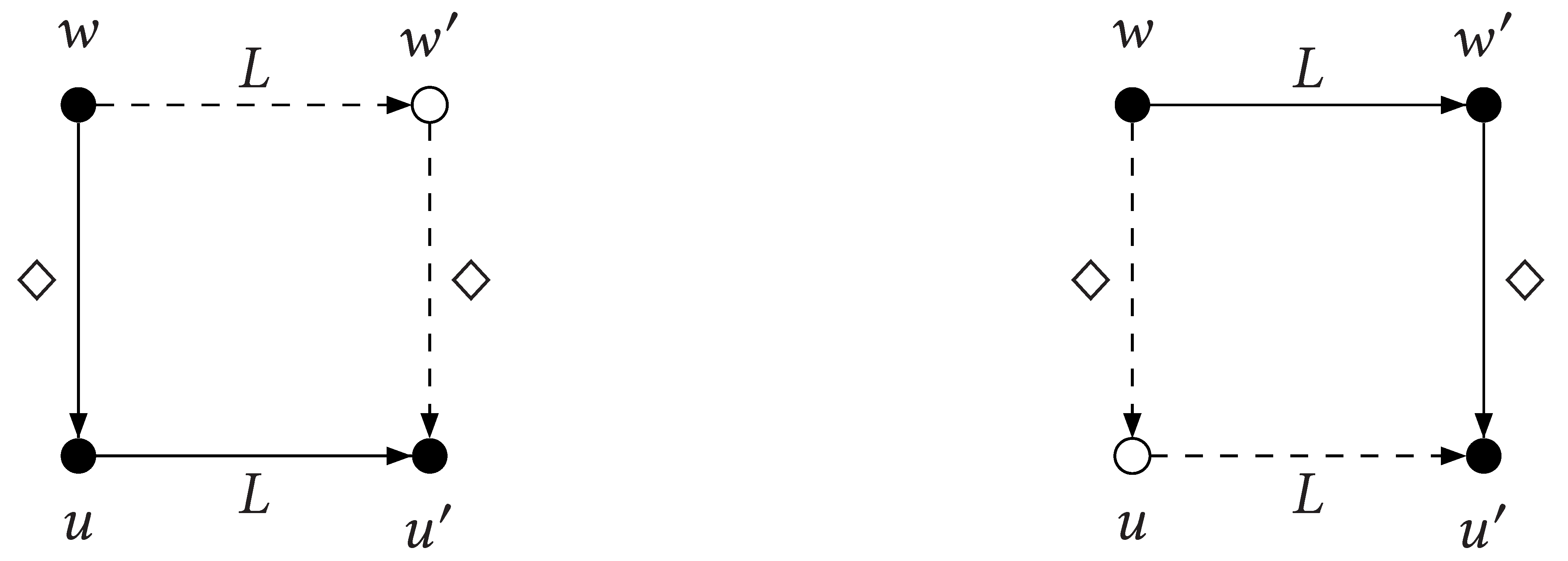}
	\caption{Left commutativity (on the left) and right commutativity (on the right).}
	\label{figure:left-and-right-commutativity}
\end{figure}

In the following subsection we define the semantics of bimodal formulas with respect to such models. For a reason that will be explained in the following subsection we shall actually not work with product models but with the following, slightly more general kinds of models.

\begin{definition}[$\kxs$- and $\sxs$-Commutator Models]
	\label{def:S4S5commutator}
	\begin{enumerate}
		\item
		\sloppy{
		A \emph{$\kxs$-commuta\-tor frame} is a triple
		$(W, \stackrel{\Diamond}{\to}, \stackrel{L}{\to})$
		such that $\stackrel{\Diamond}{\to}$ is a transitive relation on $W$,
		such that $\stackrel{L}{\to}$ is an equivalence relation on $W$,
		and such that left commutativity and right commutativity hold.
		}
		
		An \emph{$\sxs$-commutator frame} is a $\kxs$-commutator frame such that additionally the relation $\stackrel{\Diamond}{\to}$ is reflexive.
		\item
		Let $X\in\{\Kfour,\Sfour\}$.
		An {\em $X\times\Sfive$-commutator model} or short \emph{$X\times\Sfive$-model} is a
		quadruple ${(W, \stackrel{\Diamond}{\to}, \stackrel{L}{\to}, \sigma)}$
		such that the triple $(W, \stackrel{\Diamond}{\to}, \stackrel{L}{\to})$ is an
		$X\times\Sfive$-commutator frame and 
		$$\sigma:AT\to\mathcal{P}(W)$$
		is a function mapping proposition letters to subsets of $W$.		
	\end{enumerate}
\end{definition}

It is clear that any $X\times\Sfive$-product frame is an $X\times\Sfive$-commutator frame and any $X\times\Sfive$-product model is an $X\times\Sfive$-commutator model, for any $X\in\{\Kfour,\Sfour\}$.

The subset space logic $\ssl$ has been defined originally via so-called subset space models~\cite{Dabrowski1992}. As we will not use them we refrain from introducing them. Dabrowski, Moss, and Parikh~\cite{Dabrowski1992} have shown that the logic $\ssl$ can equivalently be characterized by so-called cross axiom models.

\begin{definition}[Cross Axiom Models]
\begin{enumerate}
\item
	A \emph{cross axiom frame} is a tuple 
	$$M:=(W,\stackrel{\Diamond}{\to},\stackrel{L}{\to})$$ 
	such that $W$ is a non-empty set, $\stackrel{\Diamond}{\to}$ is a preorder on $W$, 	
	$\stackrel{L}{\to}$ is an equivalence relation on $W$, and left commutativity holds.
\item
	A \emph{cross axiom model} or short \emph{$\ssl$-model} is a cross axiom frame together with a function $$\sigma:AT\to\mathcal{P}(W)$$ 
	mapping proposition letters to subsets of $W$ and satisfying the following condition 
	for all $v,w\in W$ and for all propositional variables $A$: 
	$$w\stackrel{\Diamond}{\to}v \;\rightarrow \left(w\in\sigma(A) \,\leftrightarrow\, v\in\sigma(A)\right).$$
\end{enumerate}
\end{definition}

In the context of the logic $\ssl$ the left commutativity property is usually called \emph{cross property}. 
The last condition in the previous definition is often called \emph{persistence} of propositional variables.

\subsection{Three Satisfiability Notions for Bimodal Formulas}

The semantics is defined in the same way for all considered kinds of models via the satisfaction relation $\models\,\subseteq W\times\mathcal{L}$.

\begin{definition}[Semantics]
	Let $M=(W,\stackrel{\Diamond}{\to},\stackrel{L}{\to}\sigma)$ be either some $X\!\!\times\Sfive$-commutator model, for some $X\in\{\Kfour,\Sfour\}$, or a cross axiom model. The satisfaction relation $\models\,\subseteq W\times\mathcal{L}$ is defined as follows.
	Let $w\in W$, let $A$ be an arbitrary propositional variable, and let $\varphi$ be a bimodal formula. Then
	\[\begin{array}{lllcl}
	M,w\models A                 & :\iff & w\in \sigma(A),\\
	M,w\models \neg\varphi			& :\iff & M,w\not\models \varphi,\\
	M,w\models (\varphi\wedge\psi)	& :\iff & M,w\models \varphi \text{ and } M,w\models\psi,\\
	M,w\models \Box\varphi			& :\iff	& \text{for all $v\in M$ with } 
	w\stackrel{\Diamond}{\to}v \text{ we have } M,v\models\varphi,\\
	M,w\models K\varphi			    & :\iff & \text{for all $v\in M$ with }
	w\stackrel{L}{\to}v \text{ we have } M,v\models\varphi.
	\end{array}\]
\end{definition}

When the model $M$ is clear, then we often write $w \models \varphi$ instead of 
$M,w \models \varphi$. 

\begin{lemma}\label{lemma-conditions-equivalent}
	Let $X\in\{\Kfour,\Sfour\}$.
	For a bimodal formula $\varphi \in\mathcal{L}$
	the following two conditions are equivalent.
	\begin{enumerate}
		\item
		There exist an $X\times\Sfive$-product model $M$
		and some point $w$ in $M$ such that $M,w\models\varphi$.
		\item
		There exist an $X\times\Sfive$-commutator model $M$ and some $w$ in $M$ such that $M,w\models\varphi$.
	\end{enumerate}
\end{lemma}

\begin{proof}
	The direction ``$1 \Rightarrow 2$'' is clear.
	The direction ``$2 \Rightarrow 1$'' was shown by Gabbay and Shehtman \cite[Theorem 7.12]{gabbay1998products}.
\end{proof}

\begin{definition}[$\kxs$-Satisfiable and $\sxs$-Satisfiable Formulas]
	Let $X\in\{\Kfour,\Sfour\}$.
	A bimodal formula $\varphi \in \mathcal{L}$ is \emph{$X\times\Sfive$-satisfiable} iff 
	one and then both of the two equivalent conditions in Lemma \ref{lemma-conditions-equivalent} are satisfied.
\end{definition}

Let $X\in\{\Kfour,\Sfour\}$. Actually, it is known that whenever a formula $\varphi$ is $X\times\Sfive$-satisfiable then there exists even an $X\times\Sfive$-commutator model of size doubly exponential in the length of $\varphi$ \cite[Theorem 5.27]{Kurucz2003}. This is not true for product models: there exists an $X\times\Sfive$-satisfiable formula $\varphi$ such that any $X\times\Sfive$-product model $(M, w)$ of $\varphi$ is infinite \cite[Theorem 5.32]{Kurucz2003}. This is the reason why in this article we shall work with commutator models.

\begin{definition}[$\ssl$-Satisfiable Formulas]
	A bimodal formula $\varphi \in \mathcal{L}$ is \emph{$\ssl$-satisfiable} iff 
	there exist a cross axiom model $M$ and some point $w$ in $M$ such that $M,w\models\varphi$.
\end{definition}

We already mentioned that the original definition of the subset space logic $\ssl$ was via so-called subset space models. A bimodal formula has a subset space model iff it has a cross axiom model~\cite{Dabrowski1992}. But with respect to these two kinds of models the situation is similar as above. On the one hand, for every $\ssl$-satisfiable bimodal formula there exists a finite cross axiom model, even a cross axiom model of size doubly exponential in the length of $\varphi$ \cite[Section 2.3]{Dabrowski1992}. But there are $\ssl$-satisfiable bimodal formulas that do not have a finite subset space model ~\cite[Example B]{Dabrowski1992}. This is the reason why in this article we shall work with cross axiom models.

\section{A Stronger Main Result and an Overview of the Proof}
\label{section:overview}

It is the goal of this article to show that the satisfiability problems of the three bimodal logics $\kxs$, $\sxs$, and $\ssl$ are in $\ESPACE$. Actually, we shall prove the following theorem.

\begin{theorem}
\label{theorem:upperbound}
\begin{enumerate}
\item
The satisfiability problem of the bimodal logic $\kxs$ can be decided in space $O(n\cdot 2^{3n})$.
\item
The satisfiability problems of the two bimodal logics $\sxs$ and $\ssl$ can be decided in space $O(n\cdot 2^{2n})$.
\end{enumerate}
\end{theorem}

We present decision algorithms for these problems that are based on certain kinds of tableaux.
Details about tableau methods for modal logics can be found in the following sources:
Fitting \cite{Fitting1983,MR3618504},
Gor{\'e} \cite{gore1999tableau},
Governatori \cite{DBLP:conf/aimlGovernatori08},
and Baader and Sattler \cite{MR1866611}.
The rest of the paper is organized as follows.
\begin{itemize}
	\item
	We start with some general observations about transitive relations, equivalence relations and the maximum chain length of a finite relation. Some of them will be used for formulating the algorithms, others will be used for upper estimates of the space used by the algorithms. In particular the observations about the maximum chain length might turn out to be useful in other contexts as well.
	\item 
	Then we define what we call \emph{partial tableaux}, for each of the three logics. 
	\item
	We then show that the existence of a partial tableau for a bimodal formula $\varphi$ is equivalent to its satisfiability in the respective class of models. 
	\item 
	We present recursive tableau algorithms that decide if there exists a partial tableau for a given bimodal formula $\varphi$ or not, and we prove the correctness of these algorithms. They are somewhat similar to the recursive algorithm of  Ladner~\cite{ladner1977computational} for the modal logic $\mathrm{S4}$. 
	\item 
	We show that, for $X\in \{\kxs, \sxs, \ssl\}$, given a bimodal formula $\varphi$ of length $n$ the space used by the algorithm for the logic $X$ is of the order $O(n\cdot |\mathcal{T}^X_\varphi|^3)$ where $\mathcal{T}^X_\varphi$ is the set of all so-called $X$-tableau-sets with respect to $\varphi$ (to be defined in Section~\ref{section:DefinitionTableaux}). Note that it is obvious that $|\mathcal{T}^X_\varphi| \leq 2^n$. Thus, we establish $O(n\cdot 2^{3n})$ as an upper bound for the space complexity of the satisfiability problems of all three logics. 
	\item
	Then we consider the cases $X\in\{\sxs,\ssl\}$. By an additional counting argument, we show that $|\mathcal{T}^X_\varphi| \leq 2^{2n/3}$, for all $n\geq 3$, where $\varphi$ is any bimodal formula and $n$ its length. Thus, the algorithms for $X\in\{\sxs,\ssl\}$ actually work in space $O(n\cdot 2^{2n})$. This can certainly be improved even further. A similar counting argument could be applied in the case $X=\kxs$ as well, but in order to do that one should slightly change the definition of tableau sets, and even then the gain is smaller. Therefore, this is not worked out here.
\end{itemize}
	
Due to the similarity of the three logics we can do much work in parallel for all three logics.

\section{Some Observations about Relations}
\label{section:relations}

\subsection{Transitive Relations and Equivalence Relations}

Let us consider some frame as in Subsection~\ref{subsection:models} consisting of a set and two binary relations on this set, one of them being at least transitive (and perhaps reflexive) and the other one being an equivalence relation such that  at least the left commutativity property holds. We wish to introduce some useful notions and to make some useful observations concerning this situation. We start with some preliminaries.
In the following let $W$ be a nonempty set, and let $\equiv$ be an equivalence relation on $W$.
As usual, for any $w\in W$, by 
\[ [w]_\equiv:=\{v\in W \mid w \equiv v\} \]
we denote the
$\equiv$-equivalence class of $w$, and, for any subset $A\subseteq W$, by 
\[ A_\equiv:=\{[a]_\equiv \mid a \in A\} \]
we denote
the set of $\equiv$-equivalence classes of elements of $A$.

\begin{definition}[Induced Relation]
	\label{def: induced relation}
	For any binary relation $R\subseteq W\times W$ we define the relation
	$R^\equiv\subseteq W_\equiv\times W_\equiv$
	\emph{induced on $W_\equiv$ by} $R$ by
	\[ C \,R^\equiv\, D \;:\Leftrightarrow\;
	(\exists w\in C)(\exists v\in D) \; w R v, \]
	for $C, D \in W_\equiv$.
\end{definition}

\begin{lemma}
	\label{lemma: induced relation: reflexive}
	If $R$ is reflexive then $R^\equiv$ is reflexive as well.
\end{lemma}
\begin{proof}
	Consider some $C\in W_\equiv$.
	Then $C$ is nonempty, that is, there is some $w\in C$.
	Then, as $R$ is reflexive, we have $wRw$.
	This implies $C\,R^\equiv\,C$.
	Hence, $R^\equiv$ is reflexive.
\end{proof}

\begin{lemma}
	\label{lemma: induced relation: trans}
	If $R$ is transitive and the relations $R$ and $\equiv$ 
	have the left commutativity property
	then $R^\equiv$ is transitive as well.
\end{lemma}
\begin{proof}
	Consider $C,D,E\in W_\equiv$ with 
	$C\,R^\equiv\,D$ and $D\,R^\equiv\,E$.
	We wish to show $C\,R^\equiv\,E$.
	There exist $w\in C$, $v, v'\in D$ and $u\in E$ with $wRv$ and $v'Ru$.
	Due to the left commutativity property there exists some $w'\in C$ with $w'Rv'$.
	As $R$ is transitive, we obtain $w'Ru$. Hence $C\,R^\equiv\,E$.
\end{proof}

\begin{corollary} \label{cor: induced Diamond-relation}
	Let $M=(W,\stackrel{\Diamond}{\to}, \stackrel{L}{\to})$ be a triple
	consisting of a set $W$, a transitive relation $\stackrel{\Diamond}{\to}$ on $W$ and an equivalence
	relation $\stackrel{L}{\to}$ on $W$ such that $\stackrel{\Diamond}{\to}$ and $\stackrel{L}{\to}$ have the left commutativity property.
	\begin{enumerate}
	\item
	Then $\stackrel{\Diamond}{\to}^{\stackrel{L}{\to}}$ is a transitive relation on $W_{\stackrel{L}{\to}}$.
	\item
	If the relation $\stackrel{\Diamond}{\to}$ is even a preorder then $\stackrel{\Diamond}{\to}^{\stackrel{L}{\to}}$ is a preorder as well.
	\end{enumerate}
\end{corollary}

\begin{proof}
	This follows from Lemmas~\ref{lemma: induced relation: reflexive} and~\ref{lemma: induced relation: trans}.
\end{proof}

Often, in a model as described above, we will call the $\stackrel{L}{\to}$-equivalence class of a point $w$, denoted $[w]_{\stackrel{L}{\to}}$ or shorter $[w]_L$, the \emph{cloud} of $w$. 

Finally, a word about left commutativity and right commutativity as introduced in Subsection~\ref{subsection:models}.
It gives a good intuition to think of commutativity as follows. Let $C,D$ be two clouds in a model $M$ such that $C\stackrel{\Diamond}{\to}^{\stackrel{L}{\to}}D$. Then 
\begin{enumerate}
	\item 
	$M$ has the left commutativity property iff for all $v\in D$ there is some 
	$w\in C$ with $w\stackrel{\Diamond}{\to}v$ (all points in $D$ have a father in $C$).
	\item 
	$M$ has the right commutativity property iff for all $w\in C$ there is some 
	$v\in D$ with $w\stackrel{\Diamond}{\to}v$ (all points in $C$ have a son in $D$).
\end{enumerate}

\subsection{Observations about the Maximum Chain Length}
\label{subsection:maximum-chain-length}

In this subsection we define the `maximum chain length' of a transitive relation on a finite set $S$ and prove several facts about it and in particular about the induced relation on the power set $\mathcal{P}(S)$. These observations will be used in Section~\ref{section:upper-bound-alg} when we give upper bounds for the space needed by the algorithms.

What is the maximum chain length of a transitive relation on a nonempty finite set?
Let us define this.
For any relation $\leq$ on a set $S$ let the relation $<$ on $S$ be defined by
\[ s < t :\iff (s \leq t \text{ and not } t \leq s) , \]
for any $s,t \in S$,

\begin{lemma}
\label{lemma:transitive}
Let $\leq$ be a relation on a set $S$.
\begin{enumerate}
\item
For $s,t\in S$, if $s<t$ then $s \neq t$.
\item
If $\leq$ is a transitive relation on a set $S$ then 
the relation $<$ on $S$ is transitive as well.
\end{enumerate}
\end{lemma}

\begin{proof}
Let us consider some elements $s,t\in S$ with $s<t$.
Then $s\leq t$. If $s=t$ then we would have $t\leq s$ as well, contradicting $s<t$.

Let us consider some elements  $r,s,t\in S$ with $r<s$ and $s<t$. 
Then $r \leq s$ and $s \leq t$. The transitivity of $\leq$ implies $r\leq t$.
We claim that $t \leq r$ is not true. For the sake of a contradiction, let us assume
$t \leq r$. Then the transitivity of $\leq$ implies $s \leq r$ in contradiction to $r < s$.
\end{proof}

\begin{definition}
For any transitive relation $\leq$ on a finite, nonempty set $S$ we define its
{\em maximum chain length $\mathrm{mcl}(\leq)$} to be the largest natural number $l$
such that there exists a sequence $s_0,\ldots,s_l \in S$ with
$s_i < s_{i+1}$, for all $i<l$, that is, such that
\[ s_0 < s_1 < \ldots < s_l . \]
We call such a sequence a {\em $<$-chain}.
\end{definition}

\begin{corollary}
\label{corollary:transitive}
Let $S$ be a finite nonempty set.
If $\leq$ is a transitive relation on $S$ then $\mathrm{mcl}(\leq)$
is well-defined and satisfies $\mathrm{mcl}(\leq) \leq |S|-1$.
\end{corollary}

\begin{proof}
This follows from the previous lemma.
\end{proof}

The maximum chain length of an order (a reflexive, transitive and antisymmetric relation) on a finite nonempty set is often called its {\em length} or its {\em height}; see, e.g., \cite[Page~4]{Graetzer2011} or \cite[Section 2.1]{Schroeder2016}.
In other contexts the maximum chain length plus one of a preorder on a finite nonempty set $S$
is called the {\em rank} of the finite preordered set $(S,\leq)$
(if $S$ is empty then the rank is $0$); see, e.g., \cite{Kechris1995}.
We start with two simple observations.

\begin{lemma}
\label{lemma:mcl1}
If $\leq$ is a transitive relation on a finite, nonempty set $S$ then its inverse,
the relation $(\leq)^{-1}$ on $S$ defined by
\[ s (\leq)^{-1} t :\iff t \leq s , \]
for $s,t\in S$, is a transitive relation on $S$ as well, and 
$\mathrm{mcl}((\leq)^{-1}) = \mathrm{mcl}(\leq)$.
\end{lemma}

We omit the straightforward proof.
Often, instead of $(\leq)^{-1}$ we write $\geq$.

\begin{lemma}
\label{lemma:mcl2}
If $\leq_1$ and $\leq_2$ are transitive relations on a finite, nonempty set $S$, then their intersection
$\leq_3:= \leq_1 \cap \leq_2$, that is, the relation $\leq_3$ on $S$ given by
\[ s \leq_3 t :\iff (s \leq_1 t \text{ and } s \leq_2 t) , \]
for $s,t \in S$, is a transitive relation on $S$ as well, and
\[ \mathrm{mcl}(\leq_3) \leq \mathrm{mcl}(\leq_1) + \mathrm{mcl}(\leq_2) . \]
\end{lemma}

\begin{proof}
It is clear that $\leq_3$ is a transitive relation on $S$.
For the other assertion, we observe that
\[ s <_3 t \iff (( s <_1 t \text{ and } s \leq_2 t ) \text{ or } ( s \leq_1 t \text{ and } s <_2 t ) , \]
for all $s,t \in S$. Hence, if $s_0,\ldots,s_l$ is a $<_3$-chain then with
\[ I_j := \{k \in \{0,\ldots,l-1\} ~:~ s_k <_j s_{k+1}\}, \]
for $j=1,2$, we have $\{0,\ldots,l-1\} = I_1 \cup I_2$.
The elements $s_k$ for $k\in I_1\cup\{\max(I_1)+1\}$ form a $<_1$-chain
and the elements $s_k$ for $k\in I_2\cup\{\max(I_2)+1\}$ form a $<_2$-chain.
We obtain
$\mathrm{mcl}(\leq_3) \leq \mathrm{mcl}(\leq_1) + \mathrm{mcl}(\leq_2)$.
\end{proof}

If $\leq$ is a transitive relation on a set $S$ then by
\[ s\equiv t :\iff (s=t \text{ or } (s \leq t \text{ and } t \leq s)), \]
for $s,t \in S$, an equivalence relation $\equiv$ on $S$ is defined.
If $\leq$ is reflexive a well, that is, if $\leq$ is a preorder then, for all $s,t\in S$,
\[ s\equiv t \iff (s \leq t \text{ and } t \leq s). \]
Let us assume that $\leq$ is transitive.

\begin{lemma}
\label{lemma:transitive-classes-of-at-least-two-elements}
Let $\leq$ be a transitive relation on a nonempty set $S$.
If an equivalence class $q \in S_\equiv$ contains at least two different elements then
$s \leq t$ is true for all $s,t \in q$.
\end{lemma}

\begin{proof}
Let $q\in S_\equiv$ be an equivalence class containing at least two different elements.
Let us consider some $s,t\in q$. If $s \neq t$ then $s,t\in q$ implies $s \leq t$.
If $s=t$ then, due to the fact that there is at least one element $r \in q$ with $r \neq s$,
we obtain $s \leq  r$ and $r \leq s$ and, by transitivity of $\leq$, $s\leq s$ as well.
\end{proof}

Note that in particular $s\leq s$ if $s$ is an element of an equivalence class containing at least two elements. So, the restriction of a transitive relation to the union of all equivalence classes containing at least two elements is reflexive. This lemma will turn out to be important when we estimate the space used by the algorithm that checks whether a bimodal formula is $\kxs$-satisfiable.

The following proposition is the key for our upper estimates for the maximum chain length of a certain relation on the set
$\mathcal{P}(\mathcal{T}^X_\varphi)$ where $\mathcal{T}^X_\varphi$ is the set of tableau-sets with respect to a bimodal formula $\varphi$ (this will be introduced in Section~\ref{section:DefinitionTableaux}), for any $X\in \{\kxs, \sxs, \ssl\}$.

\begin{proposition}
\label{prop:mcl3}
Let $\leq$ be a transitive relation on a finite, nonempty set $S$.
Then the relation $\leq'$ on $\mathcal{P}(S)$ defined by
\[ A \leq' B :\iff (\forall b \in B)\, (\exists a \in A) \ a \leq b , \]
for $A,B\subseteq S$, 
is transitive as well, and $\mathrm{mcl}(\leq') \leq 2  \cdot |S_\equiv| \leq 2  \cdot |S|$.
\end{proposition}

\begin{proof}
It is straightforward to see that $\leq'$ is transitive.
And it is clear that $|S_\equiv| \leq |S|$.
Let us prove $\mathrm{mcl}(\leq') \leq 2  \cdot |S_\equiv|$.
Let $A$ be a subset of $S$.
Let us call an element $a \in A$ a {\em minimal element of $A$}
if there does not exist any $b \in A$ with $b < a$.
Let $A_{\min}$ be the set of minimal elements of $A$.
Let $A_{\min,\equiv}:=(A_{\min})_\equiv$
be the set of $\equiv$-equivalence classes of elements of $A_{\min}$.
Note that
\begin{equation}
\label{eq:estimate0}
(\forall a \in A) \, (\exists a' \in A_{\min}) \ a' \leq a . 
\end{equation}
Indeed, let us consider some element $a \in A$.
If $a_0 := a$ is not an element of $A_{\min}$ then there exists some $a_1 \in A$ with $a_1 < a_0$.
If $a_1 \not\in A_{\min}$ then there exists some $a_2 \in A$ with $a_2 < a_1$. And so on.
As $S$ is finite, by Corollary~\ref{corollary:transitive}
this can be repeated only finitely often, and finally we arrive 
at some $a'\in A_{\min}$ with $a'\leq a$.

Now let also $B$ be a subset of $S$.
We claim:
\begin{equation}
\label{eq:estimate1}
 \text{if } A_{\min,\equiv} = B_{\min,\equiv} \text{ then }
  (A \leq' B \text{ and } B \leq' A). 
\end{equation}
Indeed, let us assume $A_{\min,\equiv} = B_{\min,\equiv}$.
Due to~\eqref{eq:estimate0} applied to $B$ instead of $A$,
for any $b \in B$ there exists
some $b' \in B_{\min}$ with $b' \leq b$.
Due to $A_{\min,\equiv} = B_{\min,\equiv}$ there exists some
$a \in A_{\min}$ with $a \equiv b'$. We obtain
$a \equiv b' \leq b$, hence, $a\leq b$.
This shows $A \leq' B$.
By symmetry one obtains $B \leq' A$ as well.

Next, let also $C$ be a subset of $S$
and let us assume
$A \leq' B$ and $B \leq' C$.
We claim that in this case:
\begin{equation}
\label{eq:estimate2}
 \text{if } q \in A_{\min,\equiv} \setminus B_{\min,\equiv} \text{ then } q \not\in C_{\min,\equiv}  .
\end{equation}
Let us consider some $q\in A_{\min,\equiv} \setminus B_{\min,\equiv}$.
For the sake of a contradiction, let us assume $q \in C_{\min,\equiv}$.
Fix some $a \in q \cap A_{\min}$ and some $c\in q \cap C$.
Due to $B \leq' C$, there exists some $b \in B$ with $b \leq c$.
Due to~\eqref{eq:estimate0} applied to $B$ instead of $A$,
there exists some $b' \in B_{\min}$ with $b' \leq b$.
Due to $A \leq' B$, there exists some $a' \in A$ with $a' \leq b'$.
We obtain 
\[ a' \leq b' \leq b \leq c \equiv a , \]
hence, $a' \leq a$.
Due to $a \in A_{\min}$ we conclude $a \leq a'$, and this implies
$a' \equiv b' \equiv b \equiv c \equiv a$. Hence
$q = [b']_\equiv \in B_{\min,\equiv}$ in contradiction to the assumption.
We have proved~\eqref{eq:estimate2}.

Finally, let us consider a sequence
$(A^{(0)},\ldots,A^{(l)})$ of subsets of $S$
with $A^{(i)} <' A^{(i+1)}$ for all $i<l$.
The claim \eqref{eq:estimate1} shows that for every $i<l$ 
the set $A^{(i+1)}_{\min,\equiv}$ 
is different from the set $A^{(i)}_{\min,\equiv}$.
Hence, in each step from $i$ to $i+1$ some class $q \in S_\equiv$ has to enter or to leave the set 
$A^{(\ldots)}_{\min,\equiv}$.
The claim \eqref{eq:estimate2} shows that 
once a class $q \in S_\equiv$ has left the set $A^{(\ldots)}_{\min,\equiv}$
it can never re-enter it.
Hence, any element $q \in S_\equiv$
can enter this set at most once and can leave it at most once.
This shows that this set can change at most $2 \cdot |S_\equiv|$ times.
This proves $l\leq 2 \cdot |S_\equiv|$.
\end{proof}

\section[Definition of Tableaux for $\kxs$, $\sxs$, and $\ssl$]{The Definition of Tableaux for $\kxs$, $\sxs$, and $\ssl$}
\label{section:DefinitionTableaux}

At the beginning let us have a few thoughts about the construction of a tableau for a bimodal formula $\varphi$. 
We took only $\Box$ and $K$ as primitive modal operators because it often  makes proofs shorter. But it is perhaps more understandable to talk about how to handle $\Diamond$- and $L$-formulas. These formulas are introduced as abbreviations of negated $\Box$- and $K$-formulas, respectively, and they are the ones that require appropriate successor points. So in informal descriptions we will talk about $\Diamond$- and $L$-formulas while in formal parts we only use the original operators.

We will construct tableaux not as usual brick by brick, we will instead use prefabricated parts.
\begin{itemize}
	\item 
	Instead of expanding a set of formulas step-by-step to a propositional tableau we work
	with complete {\em tableau-sets} as defined in Definition~\ref{def: tableau-sets}.1 and \ref{def: tableau-sets}.2.
	\item 
	The next step is to combine tableau-sets to sets of tableau-sets, called \emph{tableau-clouds}, under the conditions given in Definition~\ref{def: tableau-sets}.3.
\end{itemize}

Tableau-clouds are somewhat similar to mosaics~\cite{nemeti1995decidable}. 
They mirror the $\stackrel{L}{\to}$-equivalence classes in corresponding models.
The benefit of working with tableau-clouds is twofold: 
On the one hand, we only have to take care of $\Diamond$-formulas because in tableau-clouds all $L$-formulas are satisfied within the tableau-cloud.
On the other hand, demanded commutativity properties are automatically satisfied if we meet the conditions for sequences of tableau-clouds defined in Definition~\ref{def: sequence of tableau-clouds}. Commutativity is hard to guarantee if one builds tableaux from single formula sets.

The tableaux we construct are sets of tableau-clouds. We construct them recursively and pathwise. A $\Diamond$-formula may demand that there exists a suitable successor to an element in a tableau-cloud. In order to arrive at a finite tableau we will not immediately try to construct a suitable new successor tableau-cloud containing a suitable successor element but first check whether in the already constructed sequence of tableau-clouds there is a suitable one that would lead to the satisfaction of the currently considered $\Diamond$-formula. Thus, one might say that the algorithm tries to construct backwards loops whenever possible.

The backwards loops and the recursive design of the intended algorithms result in the need for \emph{partial tableaux for a sequence of tableau-clouds}. 
Assume that we have to satisfy a formula $\Diamond\chi$ occurring in some tableau-cloud $\mathcal{C}$ at some component $p$, that cannot be satisfied by a backwards loop to one of the predecessors $\mathcal{C}_0,\ldots,\mathcal{C}_{m-1}$ of $\mathcal{C}$. Then we try all tableau-clouds $\mathcal{C}'$ that contain $\chi$ in some component $q$ such that $\mathcal{C}'$ can be a successor of $\mathcal{C}$ and $p$ can be linked to $q$, until one recursive tableau search for $\mathcal{C}'$ gives a positive feedback. Because of backwards loops that might be possible, we hand over to the new instance of the algorithm not only $\mathcal{C}'$ but also the sequence $\mathcal{C}_0,\ldots \mathcal{C}_{m-1},\mathcal{C}$. Additionally we hand over the formula $\varphi$ that determines the set $\subf(\varphi)$ and the set of tableau-clouds defined below.

We speak of a \emph{partial} tableau for the sequence $(\varphi,\mathcal{C}_0,\ldots \mathcal{C}_{m-1},\mathcal{C})$ because in the present instance of the algorithm we do not care whether the elements of the sequence 
$(\mathcal{C}_0,\ldots \mathcal{C}_{m-1})$ 
can be provided with all successors needed to satisfy their $\Diamond$-formulas.
This is checked by other instances of the algorithm.

We start with the definition of tableau-sets and tableau-clouds as the building blocks of the aimed tableaux.

\begin{definition}[Tableau-sets and Tableau-clouds]
	\label{def: tableau-sets}
	Let $X\in \{\kxs,\,\sxs,\,\ssl\}$, and let $\varphi$ be a bimodal formula.
	\begin{enumerate}
		\item 
		A \emph{$\kxs$-tableau-set with respect to $\varphi$} is a subset $F\subseteq \subf(\varphi)$ such that the following conditions are satisfied for all $\psi\in \subf(\varphi)$:
		\begin{enumerate}
			\item
			If $\psi=\neg \chi$ then $(\psi \in F \iff \chi \not\in F)$.
			\item
			If $\psi = (\chi_1 \wedge \chi_2)$ then
			$(\psi \in F \iff (\chi_1 \in F \text{ and } \chi_2 \in F))$.
			\item
			If $\psi = K \chi$ then $(\psi \in F \Rightarrow \chi \in F)$.
		\end{enumerate}
		\item 
		For $X\in \{\sxs,\ssl\}$ an \emph{$X$-tableau-set with respect to $\varphi$} is a subset $F\subseteq \subf(\varphi)$ such that for all $\psi\in \subf(\varphi)$ the conditions (a), (b), and (c) of a $\kxs$-tableau-set with respect to $\varphi$ and additionally the following condition are satisfied:
		\begin{enumerate}
			\item[(d)]
			If $\psi = \Box \chi$ then $(\psi \in F \Rightarrow \chi \in F)$.
		\end{enumerate}
		\item
		The set $\mathcal{T}^X_\varphi$ of all $X$-tableau-sets with respect to $\varphi$ is defined by
		\[\mathcal{T}^X_\varphi := 
		\{F \subseteq \subf(\varphi) \mid 
	    F \text{ is an $X$-tableau-set with respect to } \varphi\}. \]
		\item 
		An \emph{$X$-tableau-cloud with respect to $\varphi$} is a subset
		$\mathcal{F} \subseteq \mathcal{T}^X_\varphi$ 
		such that the following conditions are satisfied:
		\begin{enumerate}
			\item
			For all $F,G\in \mathcal{F}$,
			$F \cap \mathcal{L}_K = G \cap \mathcal{L}_K$
			\item
			For all $\chi$ with $K\chi  \in \subf(\varphi)$, 
			if $\chi \in\bigcap_{F \in \mathcal{F}} F$
			then $K \chi \in\bigcap_{F \in \mathcal{F}} F$.
		\end{enumerate}
		\item 
		The set $\mathfrak{C}^X_\varphi$ of all $X$-tableau-clouds with respect to $\varphi$ is defined by
		\[\mathfrak{C}^X_\varphi := 
		\{\mathcal{F} \subseteq \mathcal{T}^X_\varphi \mid 
		\mathcal{F} \text{ is an $X$-tableau-cloud with respect to } \varphi\}. \]
	\end{enumerate}
\end{definition}

Before we come to the definition of tableaux we specify the conditions under which tableau-sets resp.~tableau-clouds can be composed into a \emph{sequence}.

\begin{definition}[Sequences of Tableau-sets and of Tableau-clouds] 
	\label{def: sequence of tableau-clouds} 
	\mbox{\ } \\
	Let $X\in \{\kxs,\sxs, \ssl\}$.	
	Let $\varphi$ be a bimodal formula, let $F,G\in \mathcal{T}^X_\varphi$, and let 
	$\mathcal{F},\mathcal{G}\in \mathcal{P}(\mathcal{T}^X_\varphi)$. 	
	\begin{enumerate}
		\item
		We say that $G$ \emph{can be an $X$-successor of $F$} and write shortly $F \preccurlyeq_{X} G$ if the following conditions are satisfied:
		\begin{enumerate}
		\item
		in the case $X=\kxs$ the conditions
		$$F \cap \mathcal{L}_\Box \subseteq G \text{ and } \{\psi \in \mathcal{L} \mid \Box\psi \in F\} \subseteq G,$$
		\item
		in the case $X=\sxs$ the condition
		$$F \cap \mathcal{L}_\Box \subseteq G,$$ 
		\item 
		in the case $X=\ssl$ the conditions
		$$F \cap \mathcal{L}_\Box \subseteq G \text{ and } F \cap AT = G \cap AT.$$
       \end{enumerate}
      	\item 
		We say that {\em $\mathcal{G}$ can be an $X$-successor of $\mathcal{F}$}
		and write shortly $\mathcal{F} \leq_{X} \mathcal{G}$ if the following conditions are satisfied:
		\begin{enumerate}
		\item
		in the case of $X\in\{\kxs,\sxs\}$ the two conditions
		\begin{enumerate}
			\item 
			For all $G \in \mathcal{G}$ there exists some $F\in\mathcal{F}$ such that
			$F\preccurlyeq_X G$.
			\item
			For all $F \in \mathcal{F}$ there exists some $G\in\mathcal{G}$ such that
			$F\preccurlyeq_X G.$
		\end{enumerate}
		\item
		in the case of $X=\ssl$ the condition
		\begin{enumerate}
			\item 
			For all $G \in \mathcal{G}$ there exists some $F\in\mathcal{F}$ such that
			$F \preccurlyeq_{\SSL} G.$ 
		\end{enumerate}
		\end{enumerate}
		\item 
		We define a binary relation $\equiv_X$ on $\mathcal{P}(\mathcal{T}^X_\varphi)$ by
		$$\mathcal{F} \equiv_X \mathcal{G} :\iff 
		(\mathcal{F} \leq_X \mathcal{G}  \text{ and } \mathcal{G} \leq_X \mathcal{F}).$$
		\item 
		Finally, we define a binary relation $<_X$ on $\mathcal{P}(\mathcal{T}^X_\varphi)$ by
		$$\mathcal{F} <_X \mathcal{G} :\iff 
		(\mathcal{F} \leq_X \mathcal{G}  \text{ and not } \mathcal{G} \leq_X \mathcal{F}).$$
	\end{enumerate}
\end{definition}

\begin{lemma}
	Let $\varphi$ be a bimodal formula.
	\begin{enumerate}
		\item 
		The relation $\preccurlyeq_{\kxs}$ on $\mathcal{T}^{\kxs}_\varphi$ is transitive.
		\item
		The relation $\leq_{\kxs}$ on $\mathcal{P}(\mathcal{T}^{\kxs}_\varphi)$ is transitive.
		\item
		The relation $\equiv_{\kxs}$ on $\mathcal{P}(\mathcal{T}^{\kxs}_\varphi)$ is transitive and symmetric.
	\end{enumerate}
\end{lemma}

\begin{proof}
	All assertions can be checked straightforwardly.
\end{proof}

\begin{lemma}
	Let $X\in \{\sxs,\ssl\}$, and let $\varphi$ be a bimodal formula.
	\begin{enumerate}
		\item
		The relation $\preccurlyeq_X$ on $\mathcal{T}^X_\varphi$ is a preorder.
		\item
		The relation $\leq_X$ on $\mathcal{P}(\mathcal{T}^X_\varphi)$ is a preorder.
		\item
		The relation $\equiv_X$ on $\mathcal{P}(\mathcal{T}^X_\varphi)$ is an equivalence relation.
	\end{enumerate}
\end{lemma}

\begin{proof}
	All assertions can be checked straightforwardly.
\end{proof}

\begin{definition}[Partial Tableaux for a Sequence of Tableau-clouds]
	\label{def: partial tableau}
	\mbox{\ } \\
	Let $X\in \{\kxs,\sxs,\ssl\}$, and let $\varphi$ be a bimodal formula.
	Let $(\mathcal{F}_0,\ldots,\mathcal{F}_m)$ for some $m\geq 0$ be a finite sequence
	of pairwise different $X$-tableau-clouds (that is, $\mathcal{F}_i \in \mathfrak{C}^X_\varphi$,
	for $i=0,\ldots,m$)
	with respect to $\varphi$ such that 
	$$\mathcal{F}_i \leq_X \mathcal{F}_{i+1}, \quad \text{ for all }i<m.$$
	A {\em partial $X$-tableau for $(\varphi,\mathcal{F}_0,\ldots,\mathcal{F}_m)$} is a
	subset $\mathfrak{T}\subseteq \mathfrak{C}^X_\varphi$ satisfying the following two conditions:
	\begin{enumerate}
		\item
		$\mathcal{F}_i \in \mathfrak{T}$, for $i=0,\ldots,m$.
		\item
		For all $\mathcal{F} \in \mathfrak{T} \setminus\{\mathcal{F}_0,\ldots,\mathcal{F}_{m-1}\}$,
		for all $F\in \mathcal{F}$,
		and for all $\chi$ with $\Box \chi \in \subf(\varphi)$,
		if $\Box \chi \not\in F$, then
		there exists some $\mathcal{G} \in\mathfrak{T}$ such that 
		$\mathcal{F} \leq_X \mathcal{G}$ and such that
		there exists some $G \in \mathcal{G}$ with
		$F\preccurlyeq_X G$ and 
		$\chi \not\in G$.			
	\end{enumerate}
\end{definition}

\section{Tableaux and Models}

In this section we show that the satisfiability of a bimodal formula $\varphi$ is equivalent to the existence of a partial tableau for $\varphi$. This is true for all three considered bimodal logics,
$\kxs$, $\sxs$, and $\ssl$. We proceed as follows.
\begin{itemize}
	\item
	Given a model $M$ we define for any point $w$ in $M$ the tableau-cloud ``of the point $w$''. Then we show that the set of tableau-clouds of $M$ is a partial tableau for the one-point sequence of tableau-clouds that consists of the tableau-cloud of some point $w$.
	\item
	Given a partial tableau for a one-point sequence of tableau-clouds, we construct a model that satisfies the same bimodal formulas, in a certain sense.
\end{itemize}

\begin{definition}[Tableaux based on Models]
	\label{def: tableaux based on models}
	Let $X\in \{\kxs,\,\sxs,\,\ssl\}$, and let $\varphi$ be a bimodal formula.
	Let 
	$M=(W,\stackrel{\Diamond}{\to}, \stackrel{L}{\to},\sigma)$
	be an $X$-model.
	\begin{enumerate}
		\item
		For all $w\in W$ we define 
		\[ sat_\varphi(w):=\{\psi\in \subf(\varphi) \mid M,w\models\psi\}. \]
		\item
		For $q \in W_{\stackrel{L}{\to}}$ we define
		$$\mathcal{F}_q:= \{sat_\varphi(w) \mid w\in q\}.$$
		\item
		Let 
		$$\mathfrak{T}_{M,\varphi} := \{ \mathcal{F}_q \mid q \in W_{\stackrel{L}{\to}} \}.$$		
	\end{enumerate}
\end{definition}
\begin{figure}[h]
   \centering		
   \includegraphics[width=0.8\linewidth]{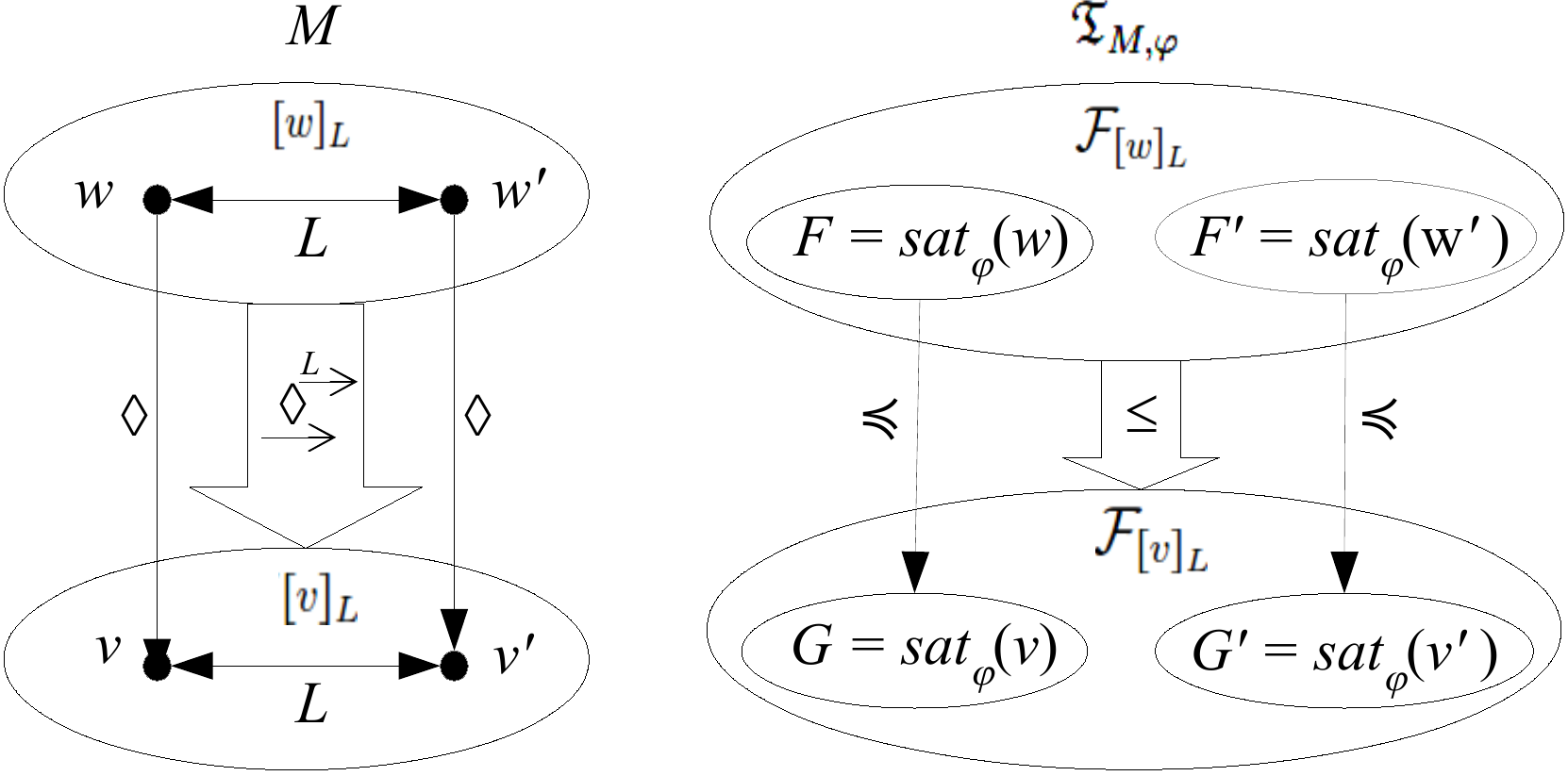}	
	\caption{An illustration of a model (on the left) and the tableau (on the right) based on it.}
	\label{figure:models-and-tableaux}
\end{figure}

\begin{lemma}
	\label{lemma:tableau-model}
	Let $X\in \{\kxs,\,\sxs,\,\ssl\}$. Let $\varphi$ be a bimodal formula.
	Let 
	$M=(W,\stackrel{\Diamond}{\to}, \stackrel{L}{\to},\sigma)$
	be an $X$-model.
	\begin{enumerate}
		\item
		For all $w \in W$, the set $sat_\varphi(w)$ is an $X$-tableau-set with respect to $\varphi$.
		\item
		For all $u,v\in W$, if $u \stackrel{\Diamond}{\to} v$ then
		$sat_\varphi(u) \preccurlyeq_X sat_\varphi(v)$.
		\item
		For all $q \in W_{\stackrel{L}{\to}}$ the set $\mathcal{F}_q$ is an $X$-tableau-cloud with respect to $\varphi$.
		\item
		For all $p,q \in W_{\stackrel{L}{\to}}$, if
		$p {\stackrel{\Diamond}{\to}}^{{\stackrel{L}{\to}}} q$ then
		$\mathcal{F}_p \leq_X \mathcal{F}_q$.
		\item
		For all $w\in W$, the set $\mathfrak{T}_{M,\varphi}$ is a partial $X$-tableau for 
		$(\varphi,\mathcal{F}_{[w]_L})$.
	\end{enumerate}
\end{lemma}

\begin{proof}
	\begin{enumerate}
		\item
		This is straightforward to see. Note that in the cases $X\in\{\sxs,\ssl\}$ the sets $sat_\varphi(w)$ for $w \in W$ satisfy Condition (d) in Definition~\ref{def: tableau-sets}.2 because the relation $\stackrel{\Diamond}{\to}$ in an $X$-model is reflexive.
		\item[2.-4.]
		All of these assertions are straightforward to check as well in each case for $X$.
		\item[5.]
		Let us fix some $w\in W$.
		It is clear that $\mathcal{F}_{[w]_L} \in \mathfrak{T}_{M,\varphi}$. 
		Let us fix some $\mathcal{F} \in \mathfrak{T}_{M,\varphi}$ and some $F\in\mathcal{F}$.
		Let us assume that $\chi$ is a bimodal formula with $\Box \chi \in \subf(\varphi)\setminus F$.
		We have to show that there exists some $\mathcal{G} \in \mathfrak{T}_{M,\varphi}$ 
		such that $\mathcal{F} \leq_X \mathcal{G}$ and such that there exists some
		$G\in\mathcal{G}$ with 
		$F\preccurlyeq_X G$ and 
		$\chi \not\in G$.			
		Indeed, let us fix some point $u\in W$ with $F=sat_\varphi(u)$ and 
		$\mathcal{F} = \mathcal{F}_{[u]_L}$. From $\Box\chi \not\in F=sat_\varphi(u)$
		we conclude $M,u \models \neg\Box\chi$, hence, $M,u \models \Diamond \neg\chi$.
		As $M$ is an $X$-model there exists some point $v\in W$ with
		$u \stackrel{\Diamond}{\to} v$ and $M,v \models \neg \chi$.
		Let $G:=sat_\varphi(v)$ and $\mathcal{G}:=\mathcal{F}_{[v]_L}$.
		Then $\neg \chi \in G$, hence, $\chi\not\in G$.
		Furthermore $G \in \mathcal{G}$ and $\mathcal{G} \in \mathfrak{T}_{M,\varphi}$.
		Finally, by the second assertion of this lemma, $u \stackrel{\Diamond}{\to} v$ implies
		$F=sat_\varphi(u) \preccurlyeq sat_\varphi(v) = G$.
		And it implies $[u]_L \stackrel{\Diamond}{\to}^{\stackrel{L}{\to}} [v]_L$, which,
		by the fourth assertion of this lemma, implies
		$\mathcal{F} = \mathcal{F}_{[u]_L} \leq_X \mathcal{G}_{[v]_L} = \mathcal{G}$.
		\qedhere
	\end{enumerate}
\end{proof}	

\begin{definition}[Models based on Tableaux]
	\label{def:models based on tableaux}
	Let $X\in \{\kxs,\sxs,\ssl\}$, and let $\varphi$ be a bimodal formula. 
	Let $\mathcal{F}_0$ be an $X$-tableau-cloud with respect to $\varphi$.
	Let $\mathfrak{T} \subseteq \mathfrak{C}^X_\varphi$ be a 
	a partial X-tableau for $(\varphi,\mathcal{F}_0)$.
	We define a quadruple
	\[ M_{\mathfrak{T}}=(W,\stackrel{\Diamond}{\to}, \stackrel{L}{\to}, \sigma) \]
	consisting of a nonempty set $W$, of two binary relations 
	$\stackrel{\Diamond}{\to}$ and $\stackrel{L}{\to}$ on $W$, and of a function
	$\sigma:AT\to\mathcal{P}(W)$ as follows:
	\[\begin{array}{rlll}
	W &:= &\{(\mathcal{F},F)\in \mathfrak{T}\times\mathcal{P}(\subf(\varphi)) \mid
	     F\in\mathcal{F}\}, & \\
	(\mathcal{F},F)\stackrel{\Diamond}{\to}(\mathcal{G},G) &:\iff
		& (\mathcal{F} \leq_X \mathcal{G} \text{ and } F\preccurlyeq_{X} G), & \\
	&& \quad\text{ for } (\mathcal{F},F), (\mathcal{G},G) \in W,\\
	(\mathcal{F},F)\stackrel{L}{\to}(\mathcal{G},G) &:\iff
		& \mathcal{F}=\mathcal{G}, &\\
	&& \quad\text{ for }(\mathcal{F},F),(\mathcal{G},G)\in W,\\ 
	\sigma(A) &:=&\{(\mathcal{F},F)\in W \mid A \in F\}, & \\
	&& \quad\text{ for } A\in AT.
	\end{array}\]
\end{definition}

\begin{figure}[h]
   \centering		
   \includegraphics[width=0.5\linewidth]{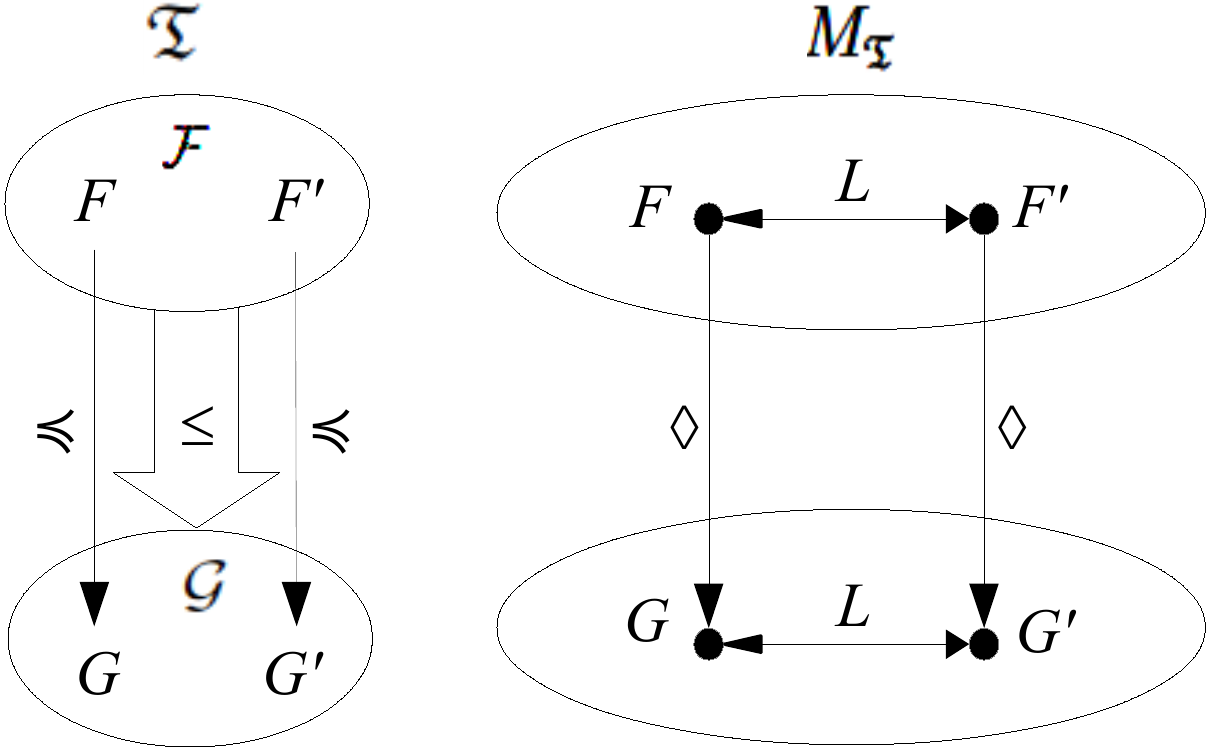}
   \caption{An illustration of a tableau (on the left) and the model (on the right) based on it.}
	\label{figure:tableaux-and-models}
\end{figure}

\begin{lemma}
	\label{lemma:model-tableau}
	Let $X$, $\varphi$, $\mathcal{F}_0$ and $\mathfrak{T}$ be as in the previous definition.
	\begin{enumerate}
		\item
		The quadruple $M_{\mathfrak{T}}$
		is an $X$-model.
		\item 
		(Truth Lemma)
		\[(\forall \psi \in \subf(\varphi))\; (\forall (\mathcal{F},F) \in W)
		\ \left( M_{\mathfrak{T}},(\mathcal{F},F) \models \psi \iff \psi \in F \right).\]
	\end{enumerate}
\end{lemma}

\begin{proof}
	\begin{enumerate}
		\item
The relations $\preccurlyeq$ on $\mathcal{T}^X_\varphi$ and 
$\leq_X$ on $\mathfrak{C}^X_\varphi$ are transitive.
Hence, the relation $\stackrel{\Diamond}{\to}$ is transitive as well.
Furthermore, in the cases $X\in\{\sxs,\ssl\}$ the relations 
$\preccurlyeq$ on $\mathcal{T}^X_\varphi$ and 
$\leq_X$ on $\mathfrak{C}^X_\varphi$ are reflexive.
Hence, in these cases the relation $\stackrel{\Diamond}{\to}$ is reflexive as well.
It is clear that the relation $\stackrel{L}{\to}$ is an equivalence relation.

Next, we show that left commutativity holds.
Let us consider pairs $(\mathcal{F},F), (\mathcal{G},G)$, $(\mathcal{G}',G') \in W$ with 
$$(\mathcal{F},F) \stackrel{\Diamond}{\to} (\mathcal{G},G) \ \text{ and }
\ (\mathcal{G},G) \stackrel{L}{\to} (\mathcal{G}',G').$$
Then $\mathcal{G}'=\mathcal{G}$.
Furthermore, $\mathcal{F} \leq_X \mathcal{G}$
and $F\preccurlyeq_X G$.
Due to $G' \in \mathcal{G}' = \mathcal{G}$ and $\mathcal{F} \leq_X \mathcal{G}$
there exists some $F' \in \mathcal{F}$ with
$F' \preccurlyeq_X G'$. We conclude
$(\mathcal{F},F') \stackrel{\Diamond}{\to} (\mathcal{G},G')$.
As $(\mathcal{F},F) \stackrel{L}{\to} (\mathcal{F},F')$ is clear,
we have shown left commutativity.

In the cases $X\in\{\kxs,\sxs\}$ right commutativity is shown in the same way.

Finally, let us consider the case $X=\ssl$.
We still need to show that in this case the persistence property holds true.
For $(\mathcal{F},F), (\mathcal{G},G) \in W$, the condition
$(\mathcal{F},F) \stackrel{\Diamond}{\to} (\mathcal{G},G)$ implies 
$F \preccurlyeq_{\ssl} G$
which, in turn, implies $F \cap AT = G \cap AT$.
Hence, for any propositional variable $A$ and any $(\mathcal{F},F), (\mathcal{G},G) \in W$
with $(\mathcal{F},F) \stackrel{\Diamond}{\to} (\mathcal{G},G)$
we have $A \in F \iff A \in G$, hence,
$(\mathcal{F},F) \in\sigma(A) \iff (\mathcal{G},G) \in\sigma(A)$.
Thus, the persistence property is satisfied.
We have shown that $M_{\mathfrak{T}}$ is a cross axiom model.
\item
		Let us consider some $\psi \in \subf(\varphi)$.
		We wish to show
		\[ M_{\mathfrak{T}},(\mathcal{F},F) \models  \psi \iff \psi \in F ,\]
		for all $(\mathcal{F},F) \in W$.		
		This is shown by structural induction.
		We distinguish the following cases:
		\begin{itemize}
			\item
			$\psi = A \in AT$.
			For $(\mathcal{F},F) \in W$,		
			the condition $M_{\mathfrak{T}},(\mathcal{F},F) \models A$ is equivalent to
			$(\mathcal{F},F)\in\sigma(A)$, and by definition of $\sigma$, this is equivalent to 
			$A \in F$.
			\item 
			$\psi = \neg \chi$.
			In this case, the following four conditions are equivalent (the second and the third condition by induction hypothesis) for $(\mathcal{F},F)\in W$:
			(a)	$M_{\mathfrak{T}},(\mathcal{F},F) \models \psi$,
			(b) $M_{\mathfrak{T}}, (\mathcal{F},F) \not\models \chi$,
			(c) $\chi \not\in F$,
			(d) $\psi \in F$.
			\item
			$\psi = \psi_1 \wedge \psi_2$. 
			This case is treated similarly.
			\item
			$\psi = K \chi$.
			Let us first assume $M_{\mathfrak{T}}, (\mathcal{F},F) \models K \chi$.
			We wish to show $K\chi \in F$.
			By the semantics definition $M_{\mathfrak{T}}, (\mathcal{F},G) \models \chi$, for all 
			$G \in \mathcal{F}$. By induction hypothesis, $\chi \in G$ for all such $G$.
			Thus, we have $\chi\in \bigcap_{G \in \mathcal{F}} G$.
			As $\mathcal{F}$ is an $X$-tableau-cloud, we obtain 
			$K\chi\in \bigcap_{G \in \mathcal{F}} G$.
			As $F \in \mathcal{F}$ as well we finally obtain $K\chi \in F$.
			
			For the other direction let us consider some $(\mathcal{F},F) \in W$,
			and let us assume $K\chi \in F$.
			We wish to show $M_{\mathfrak{T}}, (\mathcal{F},F) \models K\chi$.
			As $F \in \mathcal{F}$ and $\mathcal{F}$ is a tableau-cloud,
			we have $F \cap \mathcal{L}_K = G \cap \mathcal{L}_K$, for all $G\in \mathcal{F}$.
			This implies $K\chi \in G$, for all $G \in \mathcal{F}$. 
			As all such $G$ are $X$-tableau-sets, we obtain
			$\chi \in G$, for all $G \in \mathcal{F}$.
			By induction hypothesis $M_{\mathfrak{T}}, (\mathcal{F},G) \models \chi$,
			for all $G\in \mathcal{F}$.
			But this implies $M_{\mathfrak{T}}, (\mathcal{F},F) \models K\chi$.
			\item
			$\psi = \Box \chi$.
			Let us first assume $M_{\mathfrak{T}}, (\mathcal{F},F) \models \Box \chi$.
			We wish to show $\Box\chi \in F$.
			The assumption implies that
			$M_{\mathfrak{T}}, (\mathcal{G},G) \models \chi$, for all $(\mathcal{G},G)\in W$ with 
			$(\mathcal{F},F) \stackrel{\Diamond}{\to} (\mathcal{G},G)$.
			By induction hypothesis we obtain $\chi \in G$, for all such $(\mathcal{G},G)\in W$.
			Hence, $\chi \in G$ for all $(\mathcal{G},G) \in W$ satisfying
			$\mathcal{F} \leq_X \mathcal{G}$ and $F \preccurlyeq_X G$.
			The second condition in Definition~\ref{def: partial tableau} implies $\Box\chi\in F$.
			
			For the other direction, let us consider some $(\mathcal{F},F) \in W$ and let us assume 
			$\Box\chi \in F$.
			We wish to show $M_{\mathfrak{T}}, (\mathcal{F},F)\models \Box \chi$.
			It is sufficient to show that $M_{\mathfrak{T}}, (\mathcal{G},G) \models \chi$
			for all $(\mathcal{G},G) \in W$
			with $(\mathcal{F},F)\stackrel{\Diamond}{\to} (\mathcal{G},G)$.
			By induction hypothesis it is sufficient to show that $\chi \in G$ for all 
			$(\mathcal{G},G) \in W$ with $(\mathcal{F},F) \stackrel{\Diamond}{\to} (\mathcal{G},G)$.
			But $(\mathcal{F},F) \stackrel{\Diamond}{\to} (\mathcal{G},G)$ implies
			$F \preccurlyeq_X G$.
			In the case $X=\kxs$ this condition and the assumption $\Box\chi \in F$ immediately imply
			$\chi \in G$.
			In the cases $X\in\{\sxs,\ssl\}$ the condition $F \preccurlyeq_X G$
			and the assumption $\Box\chi \in F$ imply $\Box\chi \in G$.
			Using additionally the fact that $G$ is an $X$-tableau-set, we obtain $\chi \in G$.
			\qedhere
		\end{itemize}
\end{enumerate}
\end{proof}

We are now ready to state and prove the main result of this section.

\begin{proposition}
	\label{prop: sat equiv tablau}
	Let $X\in \{\kxs,\sxs,\ssl\}$, and let $\varphi$ be a bimodal formula.
	The following two conditions are equivalent.
		\begin{enumerate}
			\item
			$\varphi$ is $X$-satisfiable.
			\item
			There exists an $X$-tableau-cloud $\mathcal{F}_0$ such that there exist a set $F\in\mathcal{F}_0$ with $\varphi\in F$ and
			a partial $X$-tableau for $(\varphi,\mathcal{F}_0)$.
		\end{enumerate}
\end{proposition}

\begin{proof}
		Let us first assume that $\varphi$ is $X$-satisfiable.
		Then there are some $X$-model 
		$M={(W,\stackrel{\Diamond}{\to}, \stackrel{L}{\to}, \sigma)}$
		and some point $w\in W$ such that $M,w\models\varphi$.
		According to Lemma~\ref{lemma:tableau-model}.5 the set
		$\mathfrak{T}_{M,\varphi}$ defined in Definition~\ref{def: tableaux based on models} is a 
		partial $X$-tableau for $(\varphi,\mathcal{F}_{[w]_L})$.
		Due to $M,w \models \varphi$ the formula $\varphi$ is an element
		of the set $F:=sat_\varphi(w)$ and this in turn is an element of $\mathcal{F}_{[w]_L}$.
		
		For the other direction let us assume that there exist an $X$-tableau-cloud $\mathcal{F}_0$,
		an $X$-tableau-set $F\in \mathcal{F}_0$ with $\varphi \in F$ and a 	
		partial $X$-tableau $\mathfrak{T}$ for $(\varphi,\mathcal{F}_0)$.
		According to Lemma~\ref{lemma:model-tableau}.1 the quadruple
		$M_{\mathfrak{T}}=(W,\stackrel{\Diamond}{\to}, \stackrel{L}{\to}, \sigma)$
		defined in Definition~\ref{def:models based on tableaux} is an
		$X$-model. 
		Furthermore, we have $F \in \mathcal{F}_0$, hence, the
		pair $(\mathcal{F}_0,F)$ is an element of $W$.
		Finally, due to $\varphi \in F$ and due to Lemma~\ref{lemma:model-tableau}.2 we obtain
		$M_{\mathfrak{T}}, (\mathcal{F}_0,F) \models \varphi$. Hence,
		$\varphi$ is $X$-satisfiable.
\end{proof}

This shows that we can replace the search for a model of $\varphi$ by the search for a partial tableau for $\varphi$.
We will organize this search by recursive algorithms that will be described in the following section.

\section{The Tableau Algorithms}

The algorithms use the following recursive procedures
$\mathit{alg}_{\kxs}$, $\mathit{alg}_{\sxs}$, and $\mathit{alg}_{\ssl}$.

\begin{definition}[Procedures $\mathit{alg}_{\kxs}$, $\mathit{alg}_{\sxs}$, and $\mathit{alg}_{\ssl}$]
	\label{def: alg_X}
	\mbox{\ } \\
	Assume that $X\in\{\kxs, \sxs, \ssl\}$.
	Given a bimodal formula $\varphi$ and for some $m\geq 0$ a sequence
	$(\mathcal{F}_0,\ldots,\mathcal{F}_m)$  of pairwise different
	tableau-clouds $\mathcal{F}_i \in \mathfrak{C}_\varphi$
	with $\mathcal{F}_i \leq_{X} \mathcal{F}_{i+1}$, for all $i<m$,
	the algorithm
	\[ \mathit{alg}_{X}(\varphi, \mathcal{F}_0,\ldots,\mathcal{F}_m) \]
	checks for every pair $(\Box\chi,F)\in \subf(\varphi) \times \mathcal{F}_m$
	with $\Box \chi \not\in F$ first
	\begin{itemize}
		\item[(I)]
		whether there exists some $i\in\{0,\ldots,m\}$ with
		$\mathcal{F}_m \leq_{X}\mathcal{F}_i$
		and such that there exists some $G \in \mathcal{F}_i$ with 
		$F \preccurlyeq_{X} G$ and
		$\chi\not\in G$,
	\end{itemize}
	and, if this is not the case,
	\begin{itemize}
		\item[(II)]
		whether there exists some tableau-cloud
		$\mathcal{F}_{m+1} \in\mathfrak{C}^X_\varphi \setminus 
		\{\mathcal{F}_0, \ldots,\mathcal{F}_m\}$
		with $\mathcal{F}_m \leq_{X} \mathcal{F}_{m+1}$
		such that\\
		--\quad there exists some $G \in \mathcal{F}_{m+1}$ with $F \preccurlyeq_{X} G$ and $\chi\not\in G$ and \\
		--\quad 
		$\mathit{alg}_{X} (\varphi, \mathcal{F}_0,\ldots,\mathcal{F}_m,\mathcal{F}_{m+1})$ returns ``yes''.		
	\end{itemize}
	If for every pair $(\Box\chi,F)\in \subf(\varphi) \times \mathcal{F}_m$
	with $\Box \chi \not\in F$
	Condition (I) or Condition (II) is satisfied 
	then $\mathit{alg}_{X}(\varphi, \mathcal{F}_0,\ldots,\mathcal{F}_m)$
	returns ``yes'', otherwise it returns ``no''.	
	This ends the description of the algorithm $\mathit{alg}_{X}(\varphi, \mathcal{F}_0,\ldots,\mathcal{F}_m)$.
\end{definition}

We show that its works correctly, for each $X\in\{\kxs, \sxs, \ssl\}$.
 
 \begin{proposition}
	\label{prop: procedures correct}
	Let $X\in\{\kxs, \sxs, \ssl\}$. 
	Let $\varphi$ be a bimodal formula.
	Let $(\mathcal{F}_0,\ldots,\mathcal{F}_m)$ for some $m\geq 0$ be a sequence
	of pairwise different tableau-clouds with respect to $\varphi$ satisfying
	$\mathcal{F}_i \leq_X \mathcal{F}_{i+1}$, for $i<m$.
	Then $\mathit{alg}_{X}(\varphi, \mathcal{F}_0,\ldots,\mathcal{F}_m)$ returns
	``yes'' if, and only if, there exists a partial $X$-tableau for $(\varphi,\mathcal{F}_0,\ldots,\mathcal{F}_m)$.
\end{proposition}

\begin{proof}
		We show each direction of this equivalence by induction over the cardinality of the following set
		\[ S(\mathcal{F}_0,\ldots,\mathcal{F}_m):= \{\mathcal{G} \in \mathfrak{C}^X_\varphi
		\setminus\{\mathcal{F}_0, \ldots,\mathcal{F}_m\} 
		\mid \mathcal{F}_m \leq_X \mathcal{G}\}.\] 
		Note that this set is finite because $\mathfrak{C}^X_\varphi$ is a finite set.
		
		Let us first assume that there exists a partial $X$-tableau for $(\varphi,\mathcal{F}_0,\ldots,\mathcal{F}_m)$. 
		We claim that $\mathit{alg}_{X}(\varphi,\mathcal{F}_0,\ldots,\mathcal{F}_m)$ will return ``yes''.
		This is clear if there are no pairs $(\Box\chi,F)\in \subf(\varphi) \times \mathcal{F}_m$ with $\Box\chi \not\in F$,
		or if for all such pairs Condition (I) is true.
		So, let us consider the case when there are such pairs for which Condition (I) is not true.
		Let us fix a pair 
		$(\Box\chi,F)\in \subf(\varphi) \times \mathcal{F}_m$ with $\Box \chi \not\in F$
		such that (I) is not true for this pair. We claim that (II) is true for this pair.
		
		Consider a partial $X$-tableau $\mathfrak{T}$ for $(\varphi,\mathcal{F}_0,\ldots,\mathcal{F}_m)$.
		Due to $\Box \chi \in \subf(\varphi) \setminus F$ and $F\in\mathcal{F}_m$
		and due to the second condition in Definition~\ref{def: partial tableau}
		there exists an element $\mathcal{G} \in \mathfrak{T}$ with $\mathcal{F}_m \leq_X \mathcal{G}$
		such that there exists some $G \in \mathcal{G}$ with $F \preccurlyeq_X G$ and $\chi\not\in G$.
		The set $\mathcal{F}_{m+1} := \mathcal{G}$ is an $X$-tableau-cloud with 
		$\mathcal{F}_m \leq_{X} \mathcal{F}_{m+1}$, with 
		$G \in \mathcal{F}_{m+1}$, with $F\preccurlyeq_{X} G$, and with $\chi\not\in G$. 
		Furthermore, as (I) is not true for the pair $(\Box \chi, F)$, we have
		$\mathcal{F}_{m+1} \not\in \{\mathcal{F}_0,\ldots,\mathcal{F}_m)$.
		This shows that $\mathcal{F}_0,\ldots,\mathcal{F}_m,\mathcal{F}_{m+1}$ are pairwise different.
		Thus, $\mathfrak{T}$ is a partial $X$-tableau for $(\varphi,\mathcal{F}_0,\ldots,\mathcal{F}_{m+1})$.
		Due to 
		$\mathcal{F}_{m+1} \not\in \{\mathcal{F}_0, \ldots,\mathcal{F}_m\}$, the set
		$S(\mathcal{F}_0,\ldots,\mathcal{F}_m,\mathcal{F}_{m+1})$ contains
		strictly less elements than the set $S(\mathcal{F}_0,\ldots,\mathcal{F}_m)$.
		Hence, the algorithm	$\mathit{alg}_{X}(\varphi,\mathcal{F}_0,\ldots,\mathcal{F}_m,\mathcal{F}_{m+1})$
		returns ``yes'' by induction hypothesis and hence, (II) is true.
		This ends our proof by induction of the claim that
		if a partial $X$-tableau for $(\varphi,\mathcal{F}_0,\ldots,\mathcal{F}_m)$ exists
		then $\mathit{alg}_{X}(\varphi,\mathcal{F}_0,\ldots,\mathcal{F}_m)$ will return ``yes''.
		
		For the other direction, let us assume that
		$\mathit{alg}_{X}(\varphi,\mathcal{F}_0,\ldots,\mathcal{F}_m)$ returns ``yes''.
		In the following we will construct a partial $X$-tableau $\mathfrak{T}$
		for $(\varphi,\mathcal{F}_0,\ldots,\mathcal{F}_m)$.
		Let $\mathrm{Pairs}$ be the set of all pairs $(\Box\chi,F)\in \subf(\varphi) \times \mathcal{F}_m$ with 
		$\Box \chi \not\in F$. 
		As by assumption the algorithm
		$\mathit{alg}_{X}(\varphi,\mathcal{F}_0,\ldots,\mathcal{F}_m)$ returns ``yes''
		the set $\mathrm{Pairs}$ is the disjoint union of the sets
		$\mathrm{Pairs}_{I,0}, \ldots, \mathrm{Pairs}_{I,m}$, $\mathrm{Pairs}_{II}$, where
		\begin{itemize}
			\item
			$\mathrm{Pairs}_{I,i}$, for $i\in\{0,\ldots,m\}$, is the set of all pairs $(\Box\chi,F) \in \mathrm{Pairs}$
			such that (I) is satisfied and $i$ is the smallest number in $\{0,\ldots,m\}$ such that
			$\mathcal{F}_m \leq_{X}\mathcal{F}_i$
			and such that there exists some $G \in \mathcal{F}_i$ with 
			$F \preccurlyeq_{X} G$ and $\chi\not\in G$,
			\item
			$\mathrm{Pairs}_{II}$ is the set of all pairs in $\mathrm{Pairs}$ such that (I) is not satisfied but (II) is.
		\end{itemize}
		Let $k$ be the number of pairs in $\mathrm{Pairs}_{II}$, and let 
		$(\Box\chi_j,F_j)$ for $j=0,\ldots,k-1$ be the elements of $\mathrm{Pairs}_{II}$.
		For each $j\in\{0,\ldots,k-1\}$ there exists a tableau-cloud
		$\mathcal{F}_{m+1}^{(j)}\in\mathfrak{C}_\varphi \setminus \{\mathcal{F}_0, \ldots,\mathcal{F}_m\}$
		with $\mathcal{F}_m \leq_{X} \mathcal{F}_{m+1}^{(j)}$
		such that there exists some $G \in \mathcal{F}_{m+1}^{(j)}$ with 
		$F_j\preccurlyeq_{X} G$ and $\chi_j\not\in G$ and such that 
		$\mathit{alg}_{X}	(\varphi, \mathcal{F}_0,\ldots,\mathcal{F}_m,\mathcal{F}_{m+1}^{(j)})$
		returns ``yes''.
		Furthermore, the set $S(\mathcal{F}_0,\ldots,\mathcal{F}_m,\mathcal{F}_{m+1}^{(j)})$
		contains less elements than the set $S(\mathcal{F}_0,\ldots,\mathcal{F}_m)$, 
		due to $\mathcal{F}_{m+1}^{(j)} \not\in \{\mathcal{F}_0, \ldots,\mathcal{F}_m\}$.			 
		Hence, by induction hypothesis, there exists a partial $X$-tableau $\mathfrak{T}^{(j)}$
		for the sequence  $(\varphi,\mathcal{F}_0,\ldots,\mathcal{F}_m,\mathcal{F}_{m+1}^{(j)})$.
		We define
		\[ \mathfrak{T}:= \bigcup_{j=0}^{k-1} \mathfrak{T}^{(j)} . \]
		We claim that $\mathfrak{T}$ is a partial $X$-tableau for 
		$(\varphi,\mathcal{F}_0,\ldots,\mathcal{F}_m)$.
		
		Indeed, it is clear that $\{\mathcal{F}_0,\ldots,\mathcal{F}_m\} \subseteq \mathfrak{T}$
		because $\{\mathcal{F}_0,\ldots,\mathcal{F}_m\} \subseteq \mathfrak{T}^{(j)}$ even for every $j<k$.
		Let us consider some
		$\mathcal{F} \in \mathfrak{T} \setminus \{\mathcal{F}_0,\ldots,\mathcal{F}_{m-1}\}$,
		some $F \in \mathcal{F}$,
		and some formula $\Box \chi \in \subf(\varphi) \setminus F$.
		We wish to show that there exists some $\mathcal{G} \in \mathfrak{T}$ such that 
		$\mathcal{F} \leq_X \mathcal{G}$ and such that there exists some $G \in \mathcal{G}$ with
		$F\preccurlyeq_{X} G$ and $\chi \not\in G$.
		We distinguish the following two cases.
		\begin{enumerate}
			\item
			$\mathcal{F} \neq \mathcal{F}_m$. 
			Then there exists a $j\in\{0,\ldots,k-1\}$ with
			$\mathcal{F} \in \mathfrak{T}^{(j)} \setminus \{\mathcal{F}_0,\ldots,\mathcal{F}_m\}$.
			As $\mathfrak{T}^{(j)}$ is a partial $X$-tableau for
			$(\varphi,\mathcal{F}_0,\ldots,\mathcal{F}_m,\mathcal{F}_{m+1}^{(j)})$
			there exists an $X$-tableau-cloud $\mathcal{G} \in \mathfrak{T}^{(j)}$ such that 
			$\mathcal{F} \leq_X \mathcal{G}$ and such that there exists some $G \in \mathcal{G}$ with
		   $F\preccurlyeq_{X} G$ and $\chi \not\in G$.
		   As $\mathfrak{T}^{(j)}$ is a subset of $\mathfrak{T}$ we are done.
		   \item
			$\mathcal{F}=\mathcal{F}_m$.
			Then $(\Box\chi,F) \in \mathrm{Pairs}$. Either there exists a unique $i\in\{0,\ldots,m\}$ with
			$(\Box\chi,F)\in \mathrm{Pairs}_{I,i}$ or 
			$(\Box\chi,F)\in \mathrm{Pairs}_{II}$.
			
			In the first case $\mathcal{F}_m \leq_{X} \mathcal{F}_i$
			and there exists some $G \in \mathcal{F}_i$ with 
			$F\preccurlyeq_{X} G$ and
			$\chi\not\in G$. In this case we set $\mathcal{G}:=\mathcal{F}_i$.
			
			In the second case there exists a number $j\in\{0,\ldots,k-1\}$ with
			$(\Box\chi,F) = (\Box\chi_j,F_j)$.
			Then $\mathcal{F}=\mathcal{F}_m \leq_X \mathcal{F}^{(j)}_{m+1}$,
			and there exists an $X$-tableau-set $G\in \mathcal{F}^{(j)}_{m+1}$ with
			$F \preccurlyeq_X G$ and with $\chi \not\in G$.
			In this case we set $\mathcal{G}:=\mathcal{F}^{(j)}_{m+1}$.
	\end{enumerate}
	This shows that the procedure $\mathit{alg}_{X}$ is correct.
\end{proof}

Now, with the procedures $\mathit{alg}_{\kxs}$, $\mathit{alg}_{\sxs}$, and $\mathit{alg}_{\ssl}$ at hand
we can present tableau algorithms $\mathit{ALG}_{\kxs}$, $\mathit{ALG}_{\sxs}$,
and $\mathit{ALG}_{\ssl}$ for the logics under consideration.

\begin{definition}[Tableau Algorithms $\mathit{ALG}_{\kxs}$, $\mathit{ALG}_{\sxs}$, and $\mathit{ALG}_{\ssl}$]
	\label{def: ALG}
	\mbox{\ } \\
	Let $X\in\{\kxs, \sxs, \ssl\}$.
	Given a bimodal formula $\varphi$ the algorithm $\mathit{ALG}_{X}(\varphi)$
	lets $\mathcal{F}_0$ run through all $X$-tableau-clouds $\mathcal{F}_0\in\mathfrak{C}^X_\varphi$
	such that there exists some $F\in\mathcal{F}_0$
	with $\varphi\in F$ and applies $\mathit{alg}_{X}$ to 
	$(\varphi, \mathcal{F}_0)$.
	It accepts $\varphi$ 
	iff $\mathit{alg}_{X}(\varphi,\mathcal{F}_0)$ returns ``yes''
	for at least one such pair $(\varphi,\mathcal{F}_0)$.
\end{definition}

\begin{proposition}\label{prop: ALG correct}
Let $X\in\{\kxs, \sxs, \ssl\}$. The algorithm $\mathit{ALG}_X$
accepts a bimodal formula $\varphi$ if and only if $\varphi$ is $X$-satisfiable.
\end{proposition}

\begin{proof}
	Let $X\in\{\kxs, \sxs, \ssl\}$.
	Let $\varphi$ be a bimodal formula.
	The algorithm $\mathit{ALG}_{X}$ accepts $\varphi$ by definition if, and only if, there exists an $X$-tableau-cloud $\mathcal{F}_0\in\mathfrak{C}^X_\varphi$ such that $\varphi\in F$ for some $F\in\mathcal{F}_0$ and such that $\mathit{alg}_{X}(\varphi,\mathcal{F}_0)$ returns ``yes''. According to Proposition~\ref{prop: procedures correct}  $\mathit{alg}_{X}(\varphi,\mathcal{F}_0)$ returns ``yes'' if, and only if, there exists a partial tableau for $(\varphi,\mathcal{F}_0)$.
	According to Proposition~\ref{prop: sat equiv tablau} there exists a tableau-cloud
	$\mathcal{F}_0\in\mathfrak{C}^X_\varphi$
	such that there exist a set $F\in\mathcal{F}_0$ with $\varphi\in F$ and
	a partial $X$-tableau for $(\varphi,\mathcal{F}_0)$
	if, and only if, $\varphi$ is $X$-satisfiable.
\end{proof}	

Let us point out that, whenever the algorithm $\mathit{ALG}_X(\varphi)$ makes a call
$\mathit{alg}_{X}(\varphi, \mathcal{F}_0,\ldots,\mathcal{F}_m)$
for some bimodal formula $\varphi$ and some finite sequence
$\mathcal{F}_0,\ldots,\mathcal{F}_m$ of $X$-tableau-sets, then all of these
$X$-tableau-sets are pairwise different.

\section[The Space Used by the Algorithms]{Upper Bounds for the Space Used by the Algorithms}
\label{section:upper-bound-alg}

It is the purpose of this section to prove the following proposition.

\begin{proposition}
\label{prop:upperestimate}
	Let $X \in \{\kxs, \sxs, \ssl\}$.
	The algorithm $\mathit{ALG}_X$ can be implemented on a multi-tape Turing machine
	so that it, given a bimodal formula $\varphi$ of length $n$, 
	does not use more than $O(n \cdot (n+|\mathcal{T}^X_\varphi|)^3)$ space.
\end{proposition}

Before we prove this, let us deduce one of the assertions of Theorem~\ref{theorem:upperbound}

\begin{proof}[{Proof of Theorem~\ref{theorem:upperbound} in the case $X=\kxs$}]
We have presented an algorithm $\mathit{ALG}_{\kxs}$ that,
according to Proposition~\ref{prop: ALG correct}, accepts a bimodal formula $\varphi$ if, and only if, $\varphi$ is
$\kxs$-satisfiable.
Let $n$ be the length of $\varphi$.
There are at most $n$ subformulas of $\varphi$. Hence, $|\mathcal{T}^X_\varphi|\leq 2^n$.
By Proposition~\ref{prop:upperestimate} the algorithm $\mathit{ALG}_{\kxs}$
can be implemented in such a way that it works in space $O(n \cdot 2^{3\cdot n})$.
\end{proof}

In Section~\ref{section:NumberTableau-sets}, for $X \in \{\sxs,\ssl\}$
we shall give a better upper bound for $|\mathcal{T}^X_\varphi|$ than $2^n$.

Let $X\in\{\kxs,\sxs,\ssl\}$.
The algorithm $\mathit{ALG}_X$
calls the recursive procedure $\mathit{alg}_X$.
It is clear that the space used by these algorithms is heavily influenced by the
recursion depth of calls $\mathit{alg}_X(\varphi,\mathcal{F}_0,\ldots,\mathcal{F}_m)$
that occur during the execution of $\mathit{ALG}_X(\varphi)$.
Therefore, first we plan to give upper bounds for the recursion depth of these algorithms.
As a first step for this we will give
upper bounds for the maximum chain length of the transitive relation $\leq_X$ 
on $\mathcal{P}(\mathcal{T}^X_\varphi)$,
for any bimodal formula $\varphi$.

\begin{corollary}
\label{corollary:mcl-s4s5-ssl}
Let $X \in\{\kxs, \sxs, \ssl\}$, and let $\varphi$ be a bimodal formula.
Then for the relation $\leq_X$ on $\mathcal{P}(\mathcal{T}^X_\varphi)$
the following estimate is true.
\begin{enumerate}
	\item
	$\mathrm{mcl}(\leq_X) \leq 4 \cdot |\mathcal{T}^X_\varphi|$, if $X\in\{\kxs, \sxs\}$.
	\item
	$\mathrm{mcl}(\leq_{\SSL}) \leq 2 \cdot |\mathcal{T}^{\ssl}_\varphi|$.
\end{enumerate}
\end{corollary}

\begin{proof}
For $X \in \{\kxs, \sxs\}$ the relation $\leq_X$ is equal to the intersection
of the relations $\preccurlyeq_X'$ and $(\succcurlyeq_X')^{-1}$
(where with $\succcurlyeq_X$ we mean the relation
$(\preccurlyeq_X)^{-1}$, and for a relation $\leq$ the relation $\leq'$ is defined as in Proposition~\ref{prop:mcl3}).
We obtain
\begin{align*}
   \mathrm{mcl}(\leq_X)
    & \leq  \mathrm{mcl}(\preccurlyeq_X') + \mathrm{mcl}((\succcurlyeq_X')^{-1})
           && (\text{by Lemma~\ref{lemma:mcl2}}) \\
    & = \mathrm{mcl}(\preccurlyeq_X') + \mathrm{mcl}(\succcurlyeq_X')
           && (\text{by Lemma~\ref{lemma:mcl1}}) \\
    & \leq 2 \cdot |\mathcal{T}^X_\varphi| + 2 \cdot |\mathcal{T}^X_\varphi| 
           && (\text{by Prop.~\ref{prop:mcl3}}) \\
    & = 4 \cdot |\mathcal{T}^X_\varphi| .    
\end{align*}
The relation $\leq_{\SSL}$ is equal to the relation $\preccurlyeq_{\SSL}'$.
Similarly as above we obtain
$\mathrm{mcl}(\leq_{\SSL}) \leq 2 \cdot |\mathcal{T}^{\ssl}_\varphi|$.
\end{proof}

Let $X\in\{\kxs, \sxs, \ssl\}$. The following proposition contains our estimate for the recursion depth
that can occur when $\mathit{ALG}_{X}(\varphi)$ calls
the recursive procedure $\mathit{alg}_X$.

\begin{proposition}
\label{prop:recursiondepth}
Let $X\in\{\kxs, \sxs, \ssl\}$.
Let $\varphi$ be a bimodal formula. Let $n$ be its length.
	If $(\mathcal{F}_0,\ldots,\mathcal{F}_l)$ for some $l\geq 0$ is a sequence of
	$X$-tableau-clouds with respect to $\varphi$ such that
	during the execution of $\mathit{ALG}_{X}(\varphi)$ 
	a call	$\mathit{alg}_{X}(\varphi,\mathcal{F}_0,\ldots,\mathcal{F}_l)$ occurs
	then $l< 5 \cdot n \cdot |\mathcal{T}^X_\varphi|^2$.
\end{proposition}

\begin{proof}
Let $X\in\{\kxs, \sxs, \ssl\}$.
Let us assume that during the execution of $\mathit{ALG}_{X}(\varphi)$
a call $\mathit{alg}_{X}(\varphi,\mathcal{F}_0,\ldots,\mathcal{F}_l)$ occurs.
Then, during the execution of $\mathit{ALG}_{X}(\varphi)$,
for all $m \leq l$ a call $\mathit{alg}_{X}(\varphi,\mathcal{F}_0,\ldots,\mathcal{F}_m)$ must occur.
For all $m < l$ there must exist a pair
$(\Box\chi_m,F_m)\in \subf(\varphi) \times \mathcal{F}_m$ with $\Box \chi_m \not\in F_m$ 
which during the execution of $\mathit{alg}_{X}(\varphi,\mathcal{F}_0,\ldots,\mathcal{F}_m)$
leads to a call of  $\mathit{alg}_{X}(\varphi,\mathcal{F}_0,\ldots,\mathcal{F}_{m+1})$,
hence, such that,
on the one hand,
\begin{itemize}
	\item
	(I) is not satisfied, that is, there does not exist an $i\in\{0,\ldots,m\}$ with
	$\mathcal{F}_m \leq_{X}\mathcal{F}_i$
	and such that there exists some $G \in \mathcal{F}_i$ with 
	$F_m \preccurlyeq_{X} G$ and
	$\chi_m\not\in G$,
\end{itemize}
and on the other hand,
\begin{itemize}
	\item
	at least the first part of (II) is satisfied, that is,
	$\mathcal{F}_{m+1} \in\mathfrak{C}^X_\varphi \setminus 
	\{\mathcal{F}_0, \ldots,\mathcal{F}_m\}$
	and $\mathcal{F}_m \leq_{X} \mathcal{F}_{m+1}$
	and there exists some $G \in \mathcal{F}_{m+1}$ with $F_m \preccurlyeq_{X} G$ and $\chi_m\not\in G$.		
\end{itemize}
It is clear that for all $m<l$ we have $\mathcal{F}_m \leq_X \mathcal{F}_{m+1}$.
Let $m_1,\ldots,m_{k-1}$ be in increasing order the elements of the set
\[ \{j \in \{0,\ldots,l-1\} \mid \mathcal{F}_{j} <_X \mathcal{F}_{j+1} \} , \]
(this set can be empty), and set $m_0:= -1$ and $m_k:= l$. 
Then, for each $i  \in \{0,\ldots,k-1\}$, all tableau-clouds $\mathcal{F}_m$ for $m\in\{m_i+1,\ldots,m_{i+1}\}$
are pairwise $\equiv_X$-equivalent:
\[ \ldots \equiv_X \mathcal{F}_{m_i}
 <_X \mathcal{F}_{m_i+1} \equiv_X \mathcal{F}_{m_i+2} \equiv_X \ldots \equiv_X \mathcal{F}_{m_{i+1}}
 <_X \mathcal{F}_{m_{i+1}+1} \equiv_X \ldots \]
Furthermore,
\[ \mathcal{F}_{m_1} <_X < \mathcal{F}_{m_2} <_X \ldots <_X \mathcal{F}_{m_{k-1}} < \mathcal{F}_{m_k} . \]
Hence, $k-1 \leq \mathrm{mcl}(\leq_X)$.
For a moment, let us fix some $i  \in \{0,\ldots,k-1\}$.
Can there be two different numbers $m,\widetilde{m}\in\{m_i+1,\ldots,m_{i+1}\}$,
say with $m < \widetilde{m}$,
such that $(\Box\chi_m,F_m) = (\Box\chi_{\widetilde{m}},F_{\widetilde{m}})$?
We claim that this cannot be the case.
Otherwise, as at least the first part of (II) is satisfied for $m$, there is some 
$G\in \mathcal{F}_{m+1}$ with $F_m \preccurlyeq_X G$ and $\chi_m \in G$,
hence, with $F_{\widetilde{m}} \preccurlyeq_X G$ and $\chi_{\widetilde{m}} \in G$.
Furthermore, as all of the $X$-tableau-sets $\mathcal{F}_0,\ldots,\mathcal{F}_l$ are pairwise different
(this is due to the assumption that during the execution of $\mathit{ALG}_{X}(\varphi)$
a call $\mathit{alg}_{X}(\varphi,\mathcal{F}_0,\ldots,\mathcal{F}_l)$ occurs)
the set $\{\mathcal{F}_{m_i+1},\ldots, \mathcal{F}_{m_{i+1}}\}$ contains at least two different elements
(because the assumption $m,\widetilde{m}\in\{m_i+1,\ldots,m_{i+1}\}$ with $m < \widetilde{m}$,
implies that the set $\{m_i+1,\ldots,m_{i+1}\}$ contains at least two numbers), and 
by Lemma~\ref{lemma:transitive-classes-of-at-least-two-elements}
this implies $\mathcal{F}_{m+1} \leq_X \mathcal{F}_{\widetilde{m}}$
(note that $\mathcal{F}_{m+1} \leq_X \mathcal{F}_{\widetilde{m}}$ is clear if $m+1 < \widetilde{m}$ and also
if $m+1=\widetilde{m}$ and $X\in\{\sxs,\ssl\}$; Lemma~\ref{lemma:transitive-classes-of-at-least-two-elements}
is needed only for the case $m+1=\widetilde{m}$ and $X=\kxs$).
But these facts together would contradict the fact that (I) is not satisfied for $\widetilde{m}$.
We conclude that for pairwise different numbers $m,\widetilde{m}\in\{m_i+1,\ldots,m_{i+1}\}$
we have $(\Box\chi_m,F_m) \neq (\Box\chi_{\widetilde{m}},F_{\widetilde{m}})$.
This implies
\[ m_{i+1} - m_i \leq |\subf_\Box(\varphi) \times \mathcal{T}^X_\varphi| 
   \leq (n-1) \cdot |\mathcal{T}^X_\varphi|  . \]
As this is true for all $i\in\{0,\ldots,k-1\}$, we obtain, using
Corollary~\ref{corollary:mcl-s4s5-ssl}, in all three cases for $X \in \{\kxs, \sxs, \ssl\}$,
\begin{eqnarray*}
 l &=& m_k \\
   &=& -1 + \sum_{i=0}^{k-1} (m_{i+1} - m_i) \\
   &\leq & -1 + k \cdot (n-1) \cdot |\mathcal{T}^X_\varphi|  \\
   &\leq & -1 + (\mathrm{mcl}(\leq_X)+1) \cdot (n-1) \cdot |\mathcal{T}^X_\varphi|  \\
   &\leq & -1 + (4 \cdot |\mathcal{T}^X_\varphi| + 1) \cdot (n-1) \cdot |\mathcal{T}^X_\varphi|  \\
   &< & 5 \cdot n \cdot |\mathcal{T}^X_\varphi|^2 .
\end{eqnarray*}
\end{proof}

We are now prepared for the proof of the statement formulated at the beginning.

\begin{proof}[{Proof of Proposition~\ref{prop:upperestimate}}]
Let $X\in\{\kxs, \sxs, \ssl\}$.
Before we can analyze the space used by the algorithms $\mathit{ALG}_X(\varphi)$
and $\mathit{alg}_X(\varphi,\mathcal{F}_0,\ldots,\mathcal{F}_l)$,
we have to explain how the formulas, the tableau-sets and the tableau-clouds with which these
algorithms deal are stored in a Turing machine.

Let $\varphi$ be a bimodal formula.
Let $n$ be its length
(as a string over the alphabet $\{(,),\neg,\wedge,\Box,K,$ $X,0,1\}$; compare Definition~\ref{def:syntax},
but see also Remark~\ref{remark:simplified-length}).
Let $a := |\subf(\varphi)|$ be the number of subformulas of $\varphi$. Then $a \leq n$.
Let $\psi_1,\ldots,\psi_a$ be the subformulas of $\varphi$ in some order.
We can identify any subset $T \subseteq \subf(\varphi) = \{\psi_1,\ldots,\psi_a\}$,
in particular any $X$-tableau-set, with a binary string
$s_1\ldots s_a \in\{0,1\}^a$ by defining 
\[ s_i =1 :\iff \psi_i \in T . \]
Let $A := |\mathcal{T}^X_\varphi|$ be the number of all $X$-tableau-sets with respect to $\varphi$.
Then $A \leq 2^a \leq 2^n$. In Section~\ref{section:NumberTableau-sets}
we shall give a better upper estimate of $A$ in the cases
$X\in\{\sxs,\ssl\}$.
As a preliminary step at the beginning of $\mathit{ALG}_X(\varphi)$
we can check for all binary strings $s_1\ldots s_a \in \{0,1\}^a$ in alphabetical order
whether they describe subsets
of $\subf(\varphi)$ that are $X$-tableau-sets and write down only those.
Then we obtain a list of $A$ binary strings of length $a$.
This can be considered as an alphabetical list of all $X$-tableau-sets with respect to $\varphi$.
We will keep this list stored on a working tape of the Turing machine during the
whole computation.
Note that all this can be done in space $O(a \cdot A)$.

Now any set $\mathcal{F}$
whose elements are $X$-tableau-sets with respect to $\varphi$ (so, in particular
any $X$-tableau-cloud with respect to $\varphi$) can be described in a similar
manner by a binary string $b_1\ldots b_A$ of length $A$
where
\[ b_i = 1 :\iff \text{the $i$-th $X$-tableau-set with respect to $\varphi$ is an element of } \mathcal{F} . \]
In the algorithm we will assume that any $X$-tableau-cloud is described by such
a binary string of length $A$.

Note that, given a binary string of length $A$,
it is straightforward to check whether the set of $X$-tableau-sets with respect to $\varphi$
described by this string is an $X$-tableau-cloud with respect to $\varphi$ or not, and this can also be done
within space $O(a \cdot A)$.

Let us consider the for-loop in the algorithm $\mathit{ALG}_X(\varphi)$
as defined in Definition~\ref{def: ALG}:
\begin{quote}
   the algorithm $\mathit{ALG}_{X}(\varphi)$
	lets $\mathcal{F}_0$ run through all $X$-tableau-clouds $\mathcal{F}_0\in\mathfrak{C}^X_\varphi$
	such that there exists some $F\in\mathcal{F}_0$
	with $\varphi\in F$ and applies $\mathit{alg}_{X}$ to 
	$(\varphi, \mathcal{F}_0)$.
\end{quote}
In a detailed implementation of this for-loop (``through all $X$-tableau-clouds $\mathcal{F}_0\in\mathfrak{C}^X_\varphi$
	such that there exists some $F\in\mathcal{F}_0$
	with $\varphi\in F$'')
one can run through all binary strings of length $A$
and discard all those that do not describe an $X$-tableau-cloud with respect to $\varphi$
and all those that do not contain an $X$-tableau-set $F$ with $\varphi \in F$.
It is clear that the conditions that need to be checked here can be checked in space 
$O(a \cdot A)$.

We come to the recursive calls $\mathit{alg}_X(\varphi,\mathcal{F}_0,\ldots,\mathcal{F}_m)$
of the algorithm $\mathit{alg}_X$ that may occur during the
execution of $\mathit{ALG}_X(\varphi)$.
First, remember that according to Proposition~\ref{prop:recursiondepth}
we have $m < 5 \cdot n \cdot A^2$.
We claim that with each new recursive call of 		
$\mathit{alg}_{X}(\varphi,\mathcal{F}_0,\ldots,\mathcal{F}_m)$
at most an additional number of $O(n+A)$ bits need to be stored.

Indeed, one has to go through all pairs 
$(\Box\chi,F)\in \subf(\varphi) \times \mathcal{F}_m$
with $\Box \chi \not\in F$.
These pairs can be stored using $O(\log a + a) \subseteq O(n)$ bits.
Then one checks condition (I). 
 The number $i\in\{0,\ldots,m\}$ considered in (I) can be stored in $O(\log(m)) = O(n)$ bits.
And the set $G$ considered in (I) can be stored in $a \leq n$ bits as well.
When checking whether (II) is true or not one has to look for a certain
tableau-cloud $\mathcal{F}_{m+1}$. Again, this can be stored using not more
than $A$ bits. And the set $G$ considered there can be stored in $O(n)$ space again.
Thus, one does indeed not need to use more than $O(n+A)$ space
with each new recursive call of $\mathit{alg}_{X}$.

We have seen that some preliminary steps and the initial for-loop in the
algorithm $\mathit{ALG}_X(\varphi)$ can be done in space $O(a\cdot A)$.
According to Proposition~\ref{prop:recursiondepth}
the recursion depth $m$ in the recursive
calls of $\mathit{alg}_{X}(\varphi,\mathcal{F}_0,\ldots,\mathcal{F}_m)$
occuring during the computation of $\mathit{ALG}_X(\varphi)$
is at most $5 \cdot n\cdot A^2$. Finally, each recursive call 
requires at most an additional space of $O(n+A)$.
We conclude
that $\mathit{ALG}_X(\varphi)$ can be implemented in such a way that
the space used is of the order $O(n\cdot (n+A)^3)$.
\end{proof}

\begin{remark}
\label{remark:simplified-length}
All arguments in Section~\ref{section:upper-bound-alg} and Section~\ref{section:NumberTableau-sets}
go through as well if with $n$ one does not denote
the length of the bimodal formula $\varphi$ as a string over the alphabet
$\{(,),\neg,\wedge,\Box,K,$ $X,0,1\}$ but instead the ``simplified'' length of $\varphi$
as a string over the infinite alphabet
$\{(,),\neg,\wedge,$ $\Box,K\} \cup AT$.
This can also be defined as the number of symbols different from $0,1$ in 
$\varphi$ (again as a string over the alphabet
$\{(,),\neg,\wedge,\Box,K,X,0,1\}$).
\end{remark}

\section[On the Number of Tableau-sets]{On the Number of Tableau-sets}
\label{section:NumberTableau-sets}

In the previous section we have shown that
our algorithms for the satisfiability problems of the bimodal logics $\kxs$, $\sxs$, and $\ssl$
can be implemented using not more than
$O(n \cdot (n+|\mathcal{T}^X_\varphi|)^3)$ space
where $\varphi$ is the given bimodal formula, where $n$ is its length,
and where $\mathcal{T}^X_\varphi$ for $X\in \{\kxs,\sxs,\ssl\}$
is the set of $X$-tableau-sets with respect to $\varphi$.
As there are at most $n$ subformulas of $\varphi$ we obtain $|\mathcal{T}^X_\varphi|\leq 2^n$.
Thus, we have shown that the algorithms can be implemented in space
$O(n \cdot 2^{3\cdot n})$. Hence,
the satisfiability problems of the bimodal logics $\kxs$, $\sxs$, and $\ssl$ are in $\ESPACE$.

In this section we wish to slightly improve this result in the cases $X\in\{\sxs, \ssl\}$ by giving a slightly better upper bound for
$|\mathcal{T}^X_\varphi|$. By making use of the conditions that an $X$-tableau-set has to satisfy according to
Definition~\ref{def: tableau-sets}.2 we are going to show that, for all bimodal formulas of length
$n \geq 3$, 
\[ |\mathcal{T}^X_\varphi| \leq 2^{\frac{2}{3} n} . \]
In fact, we are going to show the following result.
Let $X\in\{\sxs, \ssl\}$.
For a bimodal formula $\varphi$ let
$\ell(\varphi)$ be its ``simplified length'' as considered in Remark~\ref{remark:simplified-length},
that is, $\ell(\varphi)$ is the number of symbols different from $0,1$ in 
$\varphi$ (as a string over the alphabet
$\{(,),\neg,\wedge,\Box,K,X,0,1\}$).
For $n\geq 1$ let
\begin{eqnarray*}
T(n) &:=& \max \{ |\mathcal{T}^X_\varphi| ~:~ \varphi \text{ is a bimodal formula with } \ell(\varphi) \leq n \} .
\end{eqnarray*}

\begin{proposition}
\label{prop:number-of-tableau-sets}
$\begin{array}[t]{lll}
  & T(1) &= 2 , \\
  & T(2) &= 3 , \\
 \text{for } n\geq 3, & T(n) &< 2^{(2\cdot n/3)} . 
\end{array}$
\end{proposition}

Actually, Proposition~\ref{prop:number-of-tableau-sets}
can certainly still be improved by showing an even smaller upper bound for $T(n)$. 
One can apply similar considerations in the case $X=\kxs$. But in order to gain something in that case
one should use a slightly different definition of $\kxs$-tableau-sets, and even then the gain in considerably smaller than in the
cases $X\in\{\sxs,\ssl\}$. Therefore, we refrain from treating the case $X=\kxs$ here.

\begin{proof}[Proof of Proposition~\ref{prop:number-of-tableau-sets}]
In the whole proof we consider $X\in\{\sxs,\ssl\}$. As
the $\sxs$-tableau-sets are exactly the $\ssl$-tableau-sets, that is, as
$\mathcal{T}^{\sxs}_\varphi = \mathcal{T}^{\ssl}_\varphi$ for any bimodal formula $\varphi$, in the proof
we will always suppress $X$ and, for example, simply speak about {\em tableau-sets} instead of $X$-tableau-sets and
simply write $\mathcal{T}_\varphi$ instead of $\mathcal{T}^X_\varphi$.

In addition to $T(n)$, for $n\geq 5$ we define
\begin{eqnarray*}
T_\wedge(n) &:=& \max \{ |\mathcal{T}_\varphi| ~:~ \varphi \text{ is a bimodal formula with } \ell(\varphi) \leq n \text{  and there}\\
  && \phantom{\max \{ |\mathcal{T}_\varphi| ~:~}\text{exist bimodal formulas $\chi$ and $\psi$ with } \varphi= ( \chi \wedge \psi) \} . 
\end{eqnarray*}
Note that any bimodal formula $\varphi$ of the form $(\chi \wedge \psi)$ for bimodal formulas $\chi,\psi$
satisfies $\ell(\varphi)\geq 5$.
In addition to the assertions in the proposition we claim
\[ 
 \text{for } n\geq 5, \ T_\wedge(n) < 2^{(2\cdot n/3)-1} . 
\]
This is needed for the proof of the assertions in the proposition.
We are going to show all of these assertions by induction over $n$.

If $\varphi$ is a bimodal formula with $\ell(\varphi)=1$ then $\varphi=A\in AT$.
There are exactly two tableau-sets with respect to $\varphi$: the empty set and the set $\{A\}$.
This proves the assertion for $n=1$.

Let $\varphi$ be a bimodal formula with $\ell(\varphi)=2$. There are three cases.
\begin{enumerate}
\item
$\varphi = \neg A$ where $A\in AT$. Then there are exactly two 
tableau-sets with respect to $\varphi$: the set $\{A\}$ and the set $\{\neg A\}$.
\item
$\varphi = \Box A$ where $A\in AT$. Then there are exactly three 
tableau-sets with respect to $\varphi$: the empty set, the set $\{A\}$, and the set $\{A, \Box A\}$.
\item
$\varphi = K A$ where $A\in AT$. Then there are exactly three 
tableau-sets with respect to $\varphi$: the empty set, the set $\{A\}$, and the set $\{A, K A\}$. 
\end{enumerate}
This proves the assertion for $n=2$.
In the second case we made use of the fact that if for some bimodal formula $\chi$ the
formula $\Box \chi$ is an element of a tableau-set then $\chi$ is an element of that tableau-set as well.
Similarly, in the third case we made use of the fact that if for some bimodal formula $\chi$ the
formula $K \chi$ is an element of a tableau-set then $\chi$ is an element of that tableau-set as well.
We will make use of these facts in the following cases as well.

Let us consider now a bimodal formula $\varphi$ with $n=\ell(\varphi)\geq 3$.
We distinguish several cases.
\begin{itemize}
\item
$\varphi=\neg \chi$ for some formula $\chi$.

Then for any tableau-set $T\in\mathcal{T}_\varphi$ with respect to $\varphi$ the set
$T\cap \subf(\chi)$ is a tableau-set with respect to $\chi$. And whether the formula $\neg \chi$ is an element of a given tableau-set
$T \in \mathcal{T}_\varphi$ is determined by the answer to the question whether $\chi$ is an element of $T \cap \subf(\chi)$.
Hence, $|\mathcal{T}_\varphi| = |\mathcal{T}_\chi|$.
If $\ell(\chi)=2$ then we get $|\mathcal{T}_\varphi|  = |\mathcal{T}_\chi| \leq  3 < 4 = 2^{2 \cdot 3/3}$.
If $\ell(\chi)\geq 3$ then by induction we get $|\mathcal{T}_\varphi|  = |\mathcal{T}_\chi| < 2^{2 \cdot (n-1)/3} < 2^{2 \cdot n/3}$.
\item
$\varphi=\circ \neg \chi$ for some formula $\chi$ and $\circ\in\{\Box,K\}$.

If $\ell(\varphi)=3$ then $\chi=A$ for some $A\in AT$. In that case there
are exactly three 
tableau-sets with respect to $\varphi$: the set $\{A\}$, the set $\{\neg A\}$, and the set $\{\neg A, \circ \neg A\}$. 
Note that $3 < 4 = 2^{2 \cdot 3/3}$.

If $\ell(\varphi)\geq 4$ then we claim that
$|\mathcal{T}_\varphi| \leq 2 \cdot |\mathcal{T}_\chi|$.
Indeed, if $T$ is a tableau set with respect to $\varphi$ then $T\cap \subf(\chi)$ is a a tableau set with respect to $\chi$.
The only elements in $\subf(\varphi)\setminus \subf(\chi)$ are the two formulas
$\neg \chi$ and $\circ\neg\chi$.
The question whether $\neg \chi$ is an element of $T$ or not is determined already by $T\cap \subf(\chi)$.
We have shown $|\mathcal{T}_\varphi| \leq 2 \cdot |\mathcal{T}_\chi|$.
In the case $\ell(\varphi)= 4$ we obtain $\ell(\chi)=2$, hence,
$|\mathcal{T}_\varphi| \leq 2 \cdot |\mathcal{T}_\chi| \leq 2 \cdot 3 = 6 < 2^{2 \cdot 4/3}$.
In the case $\ell(\varphi) \geq 5$ we obtain $\ell(\chi)=\ell(\varphi)-2\geq 3$, hence, by induction hypothesis,
$|\mathcal{T}_\varphi| \leq 2 \cdot |\mathcal{T}_\chi| < 2 \cdot 2^{2 \cdot (n-2)/3} < 2^{2 \cdot n/3}$.
\item
$\varphi=\circ_1 \circ_2 \neg \chi$ for some formula $\chi$ and $\circ_1,\circ_2 \in \{\Box,K\}$.

We claim that
$|\mathcal{T}_\varphi| \leq 3 \cdot |\mathcal{T}_\chi|$.
Indeed, if $T$ is a tableau set with respect to $\varphi$ then $T\cap \subf(\chi)$ is a a tableau set with respect to $\chi$.
The only elements in $\subf(\varphi)\setminus \subf(\chi)$ are the three formulas
$\neg \chi$, $\circ_2\neg\chi$, and $\circ_1\circ_2\neg\chi$.
The question whether $\neg \chi$ is an element of $T$ or not is determined already by $T\cap \subf(\chi)$.
And for the two formulas $\circ_2 \neg \chi$ and $\circ_1\circ_2 \neg \chi$ we observe that if
$\circ_1\circ_2 \neg \chi$ is an element of $T$ then so is $\circ_2 \neg \chi$.
We have shown $|\mathcal{T}_\varphi| \leq 3 \cdot |\mathcal{T}_\chi|$.

It is clear that $\ell(\varphi)=\ell(\circ_1\circ_2\neg \chi) \geq 4$.
In the case $\ell(\varphi)=4$ we obtain $\ell(\chi)=1$, hence,
$|\mathcal{T}_\varphi| \leq 3 \cdot |\mathcal{T}_\chi| \leq 3 \cdot 2 = 6 < 2^{2 \cdot 4/3}$.
In the case $\ell(\varphi)=5$ we obtain $\ell(\chi)=2$, hence,
$|\mathcal{T}_\varphi| \leq 3 \cdot |\mathcal{T}_\chi| \leq 3 \cdot 3 = 9 < 2^{2 \cdot 5/3}$.
In the case $\ell(\varphi) \geq 6$ we obtain $\ell(\chi)=\ell(\varphi)-3\geq 3$, hence, by induction hypothesis,
$|\mathcal{T}_\varphi| \leq 3 \cdot |\mathcal{T}_\chi| < 3 \cdot 2^{2 \cdot (n-3)/3} < 2^{2 \cdot n/3}$.
\item
$\varphi=\circ_1 \circ_2 \circ_3 \chi$ for some formula $\chi$ and $\circ_1,\circ_2,\circ_3 \in\{\Box,K\}$.

Again, we will use the already mentioned fact for any subformula $\circ_i \chi$ of $\varphi$: if $\circ_i \chi$ is an element
of a tableau set with respect to $\varphi$ then $\chi$ is an element of the same tableau set.

First, let us consider the cases $\ell(\varphi)=4$ and $\ell(\varphi)=5$.
If $\ell(\varphi)=4$ then $\chi=A$ for some $A\in AT$, and one checks
that there are exactly five tableau sets with respect to $\varphi$: the sets
$\emptyset$, $\{A\}$, $\{A, \circ_3 A\}$, $\{A, \circ_3 A, \circ_2\circ_3 A\}$, $\{A, \circ_3 A, \circ_2\circ_3 A, \circ_1\circ_2\circ_3 A\}$.
Note that $5 < 2^{2 \cdot 4/3}$.
Next, let us consider the case $\ell(\varphi)=5$.
Then there exists some $A\in AT$ such that either $\chi=\neg A$ or $\chi= \circ_4 A$ for some $\circ_4 \in \{\Box,K\}$.
One checks that in the first case there are again exactly five tableau sets with respect to $\varphi$
and in the second case there are exactly six tableau sets with respect to $\varphi$.
Note that $6 < 2^{2 \cdot 5/3}$.

For the case $\ell(\varphi)\geq 6$ we claim that
$|\mathcal{T}_\varphi| \leq 4 \cdot |\mathcal{T}_\chi|$.
Indeed, if $T$ is a tableau set with respect to $\varphi$ then $T\cap \subf(\chi)$ is a a tableau set with respect to $\chi$.
And for the three formulas $\circ_3 \chi$ and $\circ_2\circ_3 \chi$ and $\circ_1\circ_2\circ_3 \chi$
there are only four possibilities: (1) none of them is an element of $T$,
(2) only $\circ_3 \chi$ is an element of $T$ (3) only $\circ_2\chi$ and $\circ_2\circ_3\chi$ are elements of $T$,
(4) all three of them are elements of $T$.
We have shown $|\mathcal{T}_\varphi| \leq 4 \cdot |\mathcal{T}_\chi|$.
In the case $\ell(\varphi)\geq 6$ we obtain $\ell(\chi)=\ell(\varphi)-3\geq 3$, hence,
$|\mathcal{T}_\varphi| \leq 4 \cdot |\mathcal{T}_\chi| < 4 \cdot 2^{2 \cdot (n-3)/3} = 2^{2\cdot n/3}$.
\item
$\varphi=\circ (\chi \wedge \psi)$ for some formulas $\chi, \psi$ and $\circ \in \{\Box,K\}$.

Then $\ell(\varphi)\geq 6$ and $\ell( (\chi \wedge\psi) ) = \ell(\varphi)-1  \geq 5$.
Using the induction hypothesis for $T_\wedge(n-1))$ we obtain 
\[
   |\mathcal{T}_\varphi|
  \leq 2 \cdot |\mathcal{T}_{(\chi \wedge \psi)}| 
  < 2 \cdot 2^{(2 \cdot (n-1)/3)-1} 
  < 2^{2 \cdot n/3} .  
\]
\end{itemize}  

Finally, let us consider the case $\varphi=(\chi \wedge \psi)$, for some formulas $\chi,\psi$.
As before, let $n:=\ell(\varphi)$. Note that $n=3+\ell(\chi)+\ell(\psi)$.
It is sufficient to prove $|\mathcal{T}_\varphi| < 2^{(2 \cdot n/3)-1}$.
We observe by induction hypothesis:
\begin{align*}
  |\mathcal{T}_\varphi|
  &\leq |\mathcal{T}_\chi| \cdot |\mathcal{T}_\psi| \\
  &\leq \begin{cases}
     2 \cdot 2 = 4 < 2^{(2 \cdot 5/3)-1}& \text{if } \ell(\chi)=1 \text{ and } \ell(\psi)=1, \\
     2 \cdot 3 = 6 < 2^{(2 \cdot 6/3)-1}& \text{if } \ell(\chi)=1 \text{ and } \ell(\psi)=2, \\
     3 \cdot 2 = 6 < 2^{(2 \cdot 6/3)-1}& \text{if } \ell(\chi)=2 \text{ and } \ell(\psi)=1, \\
     3 \cdot 3 = 9 < 2^{(2 \cdot 7/3)-1}& \text{if } \ell(\chi)=2 \text{ and } \ell(\psi)=2, \\
     2 \cdot 2^{2\cdot \ell(\psi)/3} < 2^{(2 \cdot n/3)-1}& \text{if } \ell(\chi)=1 \text{ and } \ell(\psi)\geq 3, \\
     2^{2\cdot \ell(\chi)/3} \cdot 2 < 2^{(2 \cdot n/3)-1}& \text{if } \ell(\chi)\geq 3 \text{ and } \ell(\psi)=1, \\
     3 \cdot 2^{2\cdot \ell(\psi)/3} < 2^{(2 \cdot n/3)-1}& \text{if } \ell(\chi)=2 \text{ and } \ell(\psi)\geq 3, \\
     2^{2\cdot \ell(\chi)/3} \cdot 3< 2^{(2 \cdot n/3)-1}& \text{if } \ell(\chi)\geq 3 \text{ and } \ell(\psi)=2, \\
     2^{2\cdot \ell(\chi)/3} \cdot 2^{2\cdot \ell(\psi)/3} < 2^{(2 \cdot n/3)-1}& \text{if } \ell(\chi)\geq 3 \text{ and } \ell(\psi)\geq 3.
\end{cases}
\end{align*}
\end{proof}

\begin{corollary}
\label{cor:upperestimate-improved}
Let $X\in\{\sxs,\ssl\}$.
		The algorithm $\mathit{ALG}_X$ can be implemented on a multi-tape Turing machine
		so that it, given a bimodal formula $\varphi$ of length $n$, 
		does not use more than $O(n \cdot 2^{2\cdot n})$ space.
\end{corollary}

\begin{proof}
This follows immediately from Propositions~\ref{prop:upperestimate} and \ref{prop:number-of-tableau-sets}.
\end{proof}

\begin{proof}[{Proof of Theorem~\ref{theorem:upperbound} in the cases $X\in\{\sxs,\ssl\}$}]
Let us assume $X\in \{\sxs,\ssl\}$. We have presented an algorithm $\mathit{ALG}_X$ that,
according to Proposition~\ref{prop: ALG correct}, accepts a bimodal formula $\varphi$ if, and only if, $\varphi$ is
$X$-satisfiable.
And according to Corollary~\ref{cor:upperestimate-improved} the algorithm $\mathit{ALG}_X$
can be implemented in such a way that it works in space $O(n \cdot 2^{2\cdot n})$
where $n$ is the length of the input formula $\varphi$.
\end{proof}

\bibliographystyle{abbrv}
\bibliography{dis}

\end{document}